\documentclass[sn-mathphys]{sn-jnl}% Math and Physical Sciences Reference Style

 \usepackage{etoolbox}
    \makeatletter
    \patchcmd{\ps@headings}%% Regular Pages Style 
    {\hbox to \hsize{\hfill Springer Nature 2021 \LaTeX\ template\hfill}}
    {\hbox to \hsize{\hfill \hfill}}
    {}
    {}
    \patchcmd{\ps@titlepage}%% Opening Page Style   
    {\hbox to \hsize{\hfill Springer Nature 2021 \LaTeX\ template\hfill}}
    {\hbox to \hsize{\hfill \hfill}}
    {}
    {}
    \makeatother

\jyear{2021}%

\theoremstyle{thmstyleone}%
\newtheorem{theorem}{Theorem}%  meant for continuous numbers
\newtheorem{corollary}{Corollary}% 

\theoremstyle{thmstyletwo}%

\theoremstyle{thmstylethree}%
\newtheorem{definition}{Definition}%

\DeclareMathOperator*{\polylog}{poly\,\hspace{-2pt}log}
\DeclareMathOperator*{\argmin}{arg\,\hspace{-2pt}min}
\DeclareMathOperator*{\rem}{rem\,\hspace{-2pt}}

\usepackage{mathtools}

\begin{document}

\title[Embracing Off-the-Grid Samples]{Embracing Off-the-Grid Samples}


\author[1]{\fnm{Oscar L\'{o}pez}}\email{lopezo@fau.edu}

\author[2]{\fnm{\"Ozg\"ur Y\i lmaz}}\email{oyilmaz@math.ubc.ca}

\affil[1]{\orgdiv{Harbor Branch Oceanographic Institute}, \orgname{Florida Atlantic University}, \orgaddress{\street{5600 US 1 North}, \city{Fort Pierce}, \postcode{34946}, \state{Florida}, \country{U.S.A.}}}

\affil[2]{\orgdiv{Department of Mathematics}, \orgname{University of British Columbia}, \orgaddress{\street{1984 Mathematics Road}, \city{Vancouver}, \postcode{V6T 1Z2}, \state{British Columbia}, \country{Canada}}}


\abstract{Many empirical studies suggest that samples of continuous-time signals taken at locations randomly deviated from an equispaced grid (i.e., \emph{off-the-grid}) can benefit signal acquisition, e.g., undersampling and anti-aliasing. However, explicit statements of such advantages and their respective conditions are scarce in the literature. This paper provides some insight on this topic when the sampling positions are known, with grid deviations generated i.i.d. from a variety distributions. 

By solving a square-root LASSO decoder with an interpolation kernel we demonstrate the capabilities of nonuniform samples for compressive sampling, an effective paradigm for undersampling and anti-aliasing. For functions in the Wiener algebra that admit a discrete $s$-sparse representation in some transform domain, we show that $\mathcal{O}(s\polylog N)$ random off-the-grid samples are sufficient to recover an accurate $\frac{N}{2}$-bandlimited approximation of the signal. For sparse signals (i.e.,
$s \ll N$), this sampling complexity is a great reduction in comparison to
equispaced sampling where $\mathcal{O}(N)$ measurements are needed for the
same quality of reconstruction (Nyquist-Shannon sampling theorem).

We further consider noise attenuation via oversampling (relative to a desired bandwidth), a standard technique with limited theoretical understanding when the sampling positions are non-equispaced. By solving a least squares problem, we show that $\mathcal{O}(N\log N)$ i.i.d. randomly deviated samples provide an accurate $\frac{N}{2}$-bandlimited approximation of the signal with suppression of the noise energy by a factor $\sim\frac{1}{\sqrt{\log(N)}}$.}




\maketitle

\section{Introduction}
\label{intro}
\indent The Nyquist-Shannon sampling theorem is perhaps the most impactful result in the theory of signal processing, fundamentally shaping the practice of acquiring and processing data \cite{shannon,nyquist} (also attributed to Kotel'nikov \cite{kotel}, Ferrar \cite{ferrar}, Cauchy \cite{samplingsurvey}, Ogura \cite{Ogura}, E.T. and J.M. Whittaker \cite{whittaker1,whittaker2}). In this setting, typical acquisition of a continuous-time signal involves taking equispaced samples at a rate slightly higher than a prescribed frequency $\omega$ Hz in order to obtain a bandlimited approximation via a quickly decaying kernel. Such techniques provide accurate approximations of (noisy) signals whose spectral energy is largely contained in the band $[-\omega/2,\omega/2]$ \cite{samplingsurvey,samplingbook,oversampling,nonuniform}.

As a consequence, industrial signal acquisition and post-processing methods tend to be designed to incorporate uniform sampling. However, such sampling schemes are difficult to honor in practice due to physical constraints and natural factors that perturb sampling locations from the uniform grid, i.e., \emph{nonuniform} or \emph{off-the-grid} samples. In response, nonuniform analogs of the noise-free sampling theorem have been developed, where an average sampling density proportional to the highest frequency $\omega$ of the signal guarantees accurate interpolation, e.g., Landau density \cite{Landau,nonuniform,landaudensity}. However, non-equispaced samples are typically unwanted and regarded as a burden due to the extra computational cost involved in regularization, i.e., interpolating the nonuniform samples onto the desired equispaced grid.

On the other hand, many works in the literature have considered the potential benefits of deliberate nonuniform sampling \cite{best0,best,best2,non1,non2,non3,non4,non5,non6,non7,non8,non9,non10,non11,non12,non13,non14,non18,non19,non20,non21,non22,non23,non24,non25,non26,non27,non28}. Suppression of aliasing error, i.e., \emph{anti-aliasing}, is a well known advantage of randomly perturbed samples. For example, jittered sampling is a common technique for anti-aliasing that also provides a well distributed set of samples \cite{non15,non16,non17,non14}. To the best of the authors' knowledge, this phenomenon was first noticed by Harold S. Shapiro and Richard A. Silverman \cite{best0} (also by Federick J. Beutler \cite{best,best2} and implicitly by Henry J. Landau \cite{Landau}) and remained unused in applications until rediscovered in Pixar animation studios by Robert L. Cook \cite{non1}. According to our literature review, such observations remain largely empirical or arguably uninformative for applications. Closing this gap between theory and experiments would help the practical design of such widely used methodologies.

To this end, in this paper we propose a practical framework that allows us to concretely investigate the properties of randomly deviated samples for undersampling, anti-aliasing and general noise attenuation. To elaborate (see Section \ref{notation} for notation), let $\textbf{f} :[-\frac{1}{2}, \frac{1}{2})\mapsto \mathbb{C}$ be our function of interest that belongs to some smoothness class. Our goal is to obtain a uniform discretization $f\in\mathbb{C}^N$, where an estimate of $f_k = \textbf{f}(\frac{k-1}{N}-\frac{1}{2})$ will provide an accurate approximation $\textbf{f}^{\sharp}(x)$ of $\textbf{f}(x)$ for all $x\in[-\frac{1}{2}, \frac{1}{2})$. We are given noisy non-equispaced samples, $b = \tilde{f} + d \in\mathbb{C}^m$, where $\tilde{f}_k = \textbf{f}(\frac{k-1}{m}-\frac{1}{2} + \Delta_k)$ is the nonuniformly sampled signal and $d\in\mathbb{C}^m$ encompasses unwanted additive noise. In general, we will consider functions $\textbf{f}$ with support on $[-\frac{1}{2}, \frac{1}{2})$ whose periodic extension is in the Wiener algebra $A(\Omega)$ \cite{Wiener}, where by abuse of notation $\Omega$ denotes the interval $[-\frac{1}{2}, \frac{1}{2})$ and the torus $\mathbb{T}$.

To achieve undersampling and anti-aliasing, we assume our uniform signal admits a sparse (or compressible) representation along the lines of compressive sensing \cite{CSbook,introCS,CSbook2}. We say that $f$ is \emph{compressible} with respect to a transform, $\Psi\in\mathbb{C}^{N\times N}$, if there exists some $g\in\mathbb{C}^N$ such that $f = \Psi g$ and $g$ can be accurately approximated by an $s$-sparse vector ($s\leq N$). In this scenario, our methodology consists of constructing an interpolation kernel $\mathcal{S}\in\mathbb{R}^{m\times N}$ that achieves $\mathcal{S}f\approx\tilde{f}$ accurately for smooth signals, in order to define our estimate $\textbf{f}^{\sharp}(x)$ using the discrete approximation $\Psi g^{\sharp}\approx f$ where
\begin{equation}
\label{l1minS}
g^{\sharp}\coloneqq\argmin_{h\in\mathbb{C}^N}\lambda\|h\|_1+\frac{\sqrt{N}}{\sqrt{m}}\|\mathcal{S}\Psi h - b\|_2
\end{equation}
and $\lambda\geq 0$ is a parameter to be chosen appropriately. We show that for signals in the Wiener algebra and under certain distributions $\mathcal{D}$, if we have $m\sim \mathcal{O}(s\polylog(N))$ off-the-grid samples with i.i.d. deviations $\{\Delta_k\}_{k=1}^{m}\sim\mathcal{D}$ then the approximation error $\lvert \textbf{f}^{\sharp}(x)-\textbf{f}(x)\rvert$ is proportional to $\|d\|_2$, the error of the best $s$-sparse approximation of $g$, and the error of the best $\frac{N}{2}$-bandlimited approximation of $\textbf{f}$ in the Wiener algebra norm (see equation 6.1 in \cite{Wiener}). If $s\ll N$, the average sampling rate required for our result (step size $\sim\frac{1}{s\polylog(N)}$) provides a stark contrast to standard density conditions where a rate proportional to the highest frequency $\omega\sim N$, resulting in step size $\sim \frac{1}{N}$, is needed for the same bandlimited approximation. The result is among the first to formally state the anti-aliasing nature of nonuniform sampling in the context of undersampling (see Section \ref{MainResult}). 

Removing the sparse signal model, we attenuate measurement noise (i.e., \emph{denoise}) by defining $\textbf{f}^{\sharp}(x)$ using the discrete estimate
\begin{equation}
\label{LS}
f \approx f^{\sharp} := \argmin_{h\in\mathbb{C}^N}\|\mathcal{S}h - b\|_2.
\end{equation}
In this context, our main result states that $m\geq N\log(N)$ i.i.d. randomly deviated samples provide approximation error $\lvert \textbf{f}^{\sharp}(x)-\textbf{f}(x)\rvert$ proportional to the noise level ($\frac{\|d\|_2}{\sqrt{\log(N)}}$) and the error of the best $\frac{N}{2}$-bandlimited approximation of $\textbf{f}$ in the Wiener algebra norm. Thus, by nonuniform oversampling (relative to the desired $\frac{N}{2}$-bandwidth) we attenuate unwanted noise regardless of its structure. While uniform oversampling is a common noise filtration technique, our results show that general nonuniform samples also posses this denoising property (see Section \ref{1D}).

The rest of the paper is organized as follows: Section \ref{Methodology} provides a detailed elaboration of our sampling scenario, signal model and methodology. Section \ref{MainResult} showcases our results for anti-aliasing and undersampling of compressible signals while Section \ref{1D} considers noise attenuation via oversampling. A comprehensive discussion of the results and implications is presented in Section \ref{discussion}. Several experiments and computational considerations are found in Section \ref{experiments}, followed by concluding remarks in Section \ref{conclusion}. We postpone the proofs of our statements until Section \ref{mainproof}. Before proceeding to the next section, we find it best to introduce the general notation that will be used throughout. However, each subsequent section may introduce additional notation helpful in its specific context.

\subsection{Notation} 
\label{notation}
We denote complex valued functions of real variables using bold letters, e.g., $\textbf{f}:\mathbb{R}\to\mathbb{C}$. For any integer $n\in\mathbb{N}$, $[n]$ denotes the set $\{\ell\in\mathbb{N}: 1 \leq \ell\leq n\}$. For $k,\ell\in\mathbb{N}$, $b_{k}$ indicates the $k$-th entry of the vector $b$, $D_{k\ell}$ denotes the $(k,\ell)$ entry of the matrix $D$ and $D_{k*} \ (D_{*\ell})$ denotes the $k$-th row (resp. $\ell$-th column) of the matrix. We reserve $x$ to denote real variables and write the complex exponential as $\textbf{e}(x) := e^{2\pi ix}$, where $i$ is the imaginary unit. For a vector $f\in\mathbb{C}^{n}$ and $1\leq p < \infty$, $\|f\|_p := \big[\sum_{k=1}^{n}\lvert f_k\rvert^p\big]^{1/p}$ is the $p$-norm, $\|f\|_{\infty} = \max_{k\in [n]}\lvert f_k\rvert$, and $\|f\|_0$ gives the total number of non-zero elements of $f$. For a matrix $X\in\mathbb{C}^{n\times m}$, $\sigma_k(X)$ denotes the $k$-th largest singular value of $X$ and $\|X\| := \sigma_1(X)$ is the spectral norm. $A(\Omega)$ is the Wiener algebra and $H^k(\Omega)$ is the Sobolev space $W^{k,2}(\Omega)$ (with domain $\Omega$), $S^{n-1}$ is the unit sphere in $\mathbb{C}^n$, and the adjoint of a linear operator $\mathcal{A}$ is denoted by $\mathcal{A}^*$. 

\section{Assumptions and Methodology}
\label{Methodology}

In this section we develop the signal model, deviation model and interpolation kernel, in Sections \ref{SignalModel}, \ref{DevModel} and \ref{dirichletsection} respectively. This will allow us to proceed to Sections \ref{MainResult} and \ref{1D} where the main results are elaborated. We note that the deviation model (Section \ref{DevModel}) and sparse signal model at the end of Section \ref{SignalModel} only apply to the compressive sensing results in Section \ref{MainResult}. However, the \emph{sampling on the torus} assumption in Section \ref{SignalModel} does apply to the results in Section \ref{1D} as well.

\subsection{Signal Model}
\label{SignalModel}
For the results in Section \ref{MainResult} and \ref{1D}, let $\Omega = [-\frac{1}{2},\frac{1}{2})$ and let $\textbf{f}:\Omega\to\mathbb{C}$ be the function of interest to be sampled. We assume $\textbf{f}\in A(\Omega)$ with Fourier expansion
\begin{equation}
\label{Fexpansion}
\textbf{f}(x) = \sum_{\ell=-\infty}^{\infty}c_{\ell}\textbf{e}(\ell x),
\end{equation}
on $\Omega$. Note that our regularity assumption implies that 
\[
\sum_{\ell=-\infty}^{\infty}\lvert c_{\ell}\rvert<\infty,
\]
which will be crucial for our error bounds. Further, $H^{k}(\Omega)\subset A(\Omega)$ for $k\geq 1$ so that our context applies to many signals of interest.

Henceforth, let $N\in\mathbb{N}$ be odd. We denote the discretized regular data vector by $f \in \mathbb{C}^{N}$, which is obtained by sampling $\textbf{f}$ on the uniform grid $\tau=\{t_1, \cdots, t_N\} \subset \Omega$, with $t_k := \frac{k-1}{N}-\frac{1}{2}$, (which is a collection of equispaced points) so that $f_k = \textbf{f}(t_{k})$. The vector $f$ will be our discrete signal of interest to recover via nonuniform samples in order to ultimately obtain an approximation to $\textbf{f}(x)$ for all $x\in\Omega$. Similar results can be obtained in the case $N$ is even, our current assumption is adopted to simplify the exposition.

The observed nonuniform data is encapsulated in the vector $\tilde{f} \in \mathbb{C}^{m}$ with underlying unstructured grid $\tilde{\tau} = \{\tilde{t}_1, \cdots, \tilde{t}_m\} \subset \Omega$ where $\tilde{t}_k := \frac{k-1}{m}-\frac{1}{2} + \Delta_k$ is now a collection of generally non-equispaced points. The entries of the perturbation vector $\Delta\in\mathbb{R}^{m}$ define the pointwise deviations of $\tilde{\tau}$ from the equispaced grid $\{\frac{k-1}{m}-\frac{1}{2}\}_{k=1}^{m}$, where $\tilde{f}_k = \textbf{f}(\tilde{t}_{k})$. Noisy nonuniform samples are given as
\[
b = \tilde{f} + d \in \mathbb{C}^m,
\]
where the noise model, $d$, does not incorporate off-the-grid corruption. We assume that we know $\tilde{\tau}$.

In order for the expansion (\ref{Fexpansion}) to remain valid for $x\in\tilde{\tau}$, we must impose $\tilde{\tau}\subset \Omega$. This is not possible for the general deviations $\Delta$ we wish to examine, so we instead adopt the torus as our sampling domain to ensure this condition.

\emph{Sampling on the torus:} for all our results, we consider sampling schemes to be on the torus. In other words, we allow grid points $\tilde{\tau}$ to lie outside of the interval $[-\frac{1}{2},\frac{1}{2})$ but they will correspond to samples of $\textbf{f}$ within $[-\frac{1}{2},\frac{1}{2})$ via a circular wrap-around. To elaborate, if $\textbf{f}\lvert_{\Omega}(x)$ is given as
\[
\textbf{f}\lvert_{\Omega}(x) =
  \begin{cases}
                                   \textbf{f}(x) & \text{if} \ \ x\in[-\frac{1}{2},\frac{1}{2}) \\
                                   0 & \text{if} \ \ x\notin[-\frac{1}{2},\frac{1}{2}),
  \end{cases}
\]
then we define $\tilde{\textbf{f}}(x)$ as the periodic extension of $\textbf{f}\lvert_{\Omega}(x)$ to the whole line
\[
\tilde{\textbf{f}}(x) = \sum_{\ell=-\infty}^{\infty}\textbf{f}\lvert_{\Omega}(x+\ell).
\]
We now apply samples generated from our deviations $\tilde{\tau}$ to $\tilde{\textbf{f}}(x)$. Indeed, for any $\tilde{t}_k$ generated outside of $\Omega$ will have $\tilde{\textbf{f}}(\tilde{t}_k) = \textbf{f}(t^*)$ for some $t^*\in\Omega$. In this way, we avoid restricting the magnitude of the entries of $\Delta$ and the expansion (\ref{Fexpansion}) will remain valid for any nonuniform samples generated.

\emph{Sparse signal model:} For the results of Section \ref{MainResult} only, we impose a compressibility condition on $f \in \mathbb{C}^{N}$. To this end, let $\Psi\in\mathbb{C}^{N\times N}$ be a basis with $0< \sigma_N(\Psi) =: \alpha$ and $\sigma_1(\Psi) =: \beta$. We assume there exists some $g\in\mathbb{C}^N$ such that $f = \Psi g$, where $g$ can be accurately approximated by an $s\leq N$ sparse vector. To be precise, for $s\in [N]$ we define the error of best $s$-sparse approximation of $g$ as
\[
\epsilon_s(g) := \min_{\|h\|_0\leq s}\|h-g\|_1,
\]
and assume $s$ has been chosen so that $\epsilon_s(g)$ is within a prescribed error tolerance determined by the practitioner. 

In Section \ref{mainresultproof}, we will relax the condition that $\Psi$ be a basis by allowing full column rank matrices $\Psi\in\mathbb{C}^{N\times n}$ with $n\leq N$. While such transforms are not typical in compressive sensing, we argue that they may be of practical interest since our results will show that if $\Psi$ can be selected as tall matrix then the sampling complexity will solely depend on its number of columns (i.e., the smallest dimension $n$). 

The transform $\Psi$ will have to be coherent with respect to the 1D centered discrete Fourier basis $\mathcal{F}\in\mathbb{C}^{N\times N}$ (see Section \ref{dirichletsection} for definition of $\mathcal{F}$). We define the DFT-incoherence parameter as
\[
\gamma = \max_{\ell\in [N]} \sum_{k=1}^{N}\lvert\langle \mathcal{F}_{*k},\Psi_{*\ell}\rangle\rvert,
\]
which provides a uniform bound on the $\ell_1$-norm of the discrete Fourier coefficients of the columns of $\Psi$. This parameter will play a role in the sampling complexity of our result in Section \ref{MainResult}, as a metric that quantifies the smoothness of our signal of interest. We discuss $\gamma$ in detail in Section \ref{SignalModelDiscussion}, including examples for several transforms common in compressive sensing.

\subsection{Deviation Model}
\label{DevModel}
Section \ref{MainResult} will apply to random deviations $\Delta\in\mathbb{R}^{m}$ whose entries are i.i.d. with any distribution, $\mathcal{D}$, that obeys the following: for $\delta\sim\mathcal{D}$, there exists some $\theta \geq 0$ such that for all integers $0<\lvert j\rvert\leq\frac{N-1}{m}$ we have
\begin{equation}
\label{thetadef}
\frac{2N}{m}\lvert\mathbb{E}\textbf{e}(jm\delta)\rvert \leq \theta .    
\end{equation}
This will be known as our \emph{deviation model}. In our results, distributions with smaller $\theta$ parameter will require less samples and provide reduced error bounds. We postpone further discussion of the deviation model until Section \ref{DeviationDiscussion}, where we will also provide examples of deviations that fit our criteria. We note that the deviation model is most relevant when $m<N$. The case $m\geq N$ is discussed in Section \ref{1D}, which no longer requires this deviation model or the sparse signal model.

\subsection{Dirichlet Kernel}
\label{dirichletsection}

In Section \ref{MainResult} and \ref{1D}, we model our nonuniform samples via an interpolation kernel $\mathcal{S}\in\mathbb{R}^{m\times N}$ that achieves $\mathcal{S}f\approx \tilde{f}$ accurately. We consider the Dirichlet kernel defined by $\mathcal{S} = \mathcal{N}\mathcal{F}^*:\mathbb{C}^{N}\to\mathbb{C}^{m}$, where $\mathcal{F}\in\mathbb{C}^{N\times N}$ is a 1D centered discrete Fourier transform (DFT) and $\mathcal{N}\in\mathbb{C}^{m\times N}$ is a 1D centered nonuniform discrete Fourier transform (NDFT, see \cite{NFFT,NFFT2}) with normalized rows and non-harmonic frequencies chosen according to $\tilde{\tau}$. In other words, let $\tilde{N} = \frac{N-1}{2}$, then the $(k, \ell) \in [m]\times [N]$ entry of $\mathcal{N}$ is given as
\[
\mathcal{N}_{k\ell} = \frac{1}{\sqrt{N}}\textbf{e}(-\tilde{t}_k(\ell-\tilde{N}-1)).
\]
This NDFT is referred to as a nonuniform discrete Fourier transform of type 2 in \cite{NFFT2}. Thus, the action of $\mathcal{S}$ on $f\in\mathbb{C}^{N}$ can be given as follows: we first apply the centered inverse DFT to our discrete uniform data
\begin{equation}
\label{fft}
\check{f}_u := (\mathcal{F}^*f)_u = \sum_{p=1}^{N}f_p\mathcal{F}^*_{up} := \frac{1}{\sqrt{N}}\sum_{p=1}^{N}f_p\textbf{e}(t_p(u-\tilde{N}-1)), \ \ \forall u\in[N],
\end{equation}
followed by the NDFT in terms of $\tilde{\tau}$:
\begin{equation}
\label{nfft}
(\mathcal{S}f)_k := (\mathcal{N}\check{f})_k = \sum_{u=1}^{N}\check{f}_u\mathcal{N}_{ku} := \frac{1}{\sqrt{N}}\sum_{u=1}^{N} \check{f}_u\textbf{e}(-\tilde{t}_k(u-\tilde{N}-1)), \ \ \forall k\in[m].
\end{equation}
Equivalently,
\begin{align}
\label{dirichlet}
&(\mathcal{S}f)_k =\frac{1}{N}\sum_{p=1}^{N}f_p\textbf{K}(\tilde{t}_{k}-t_p) \ \ \ \ \forall k\in[m],
\end{align}
where $\textbf{K}(x) = \frac{\sin{(N\pi x})}{\sin{(\pi x)}}$ is the Dirichlet kernel. This equality is
well known and holds by applying the geometric series formula upon expansion. This kernel is commonly used for trigonometric interpolation and is accurate when acting on signals that can be well approximated by trigonometric polynomials of finite order, as we show in the following theorem.

\begin{theorem}
\label{thmS}
Let $\mathcal{S}, f$ and $\tilde{f}$ be defined as above and $\tilde{t}_k\in\Omega$ for some $k\in [m]$. If $\tilde{t}_k = t_p$ for some $p\in [N]$ then
\begin{equation}
\label{perfectinterpolation}
\left(\tilde{f}-\mathcal{S}f\right)_k = 0
\end{equation}
and otherwise
\begin{equation}
\label{interpolationerror}
\left(\tilde{f}-\mathcal{S}f\right)_k = \sum_{\lvert\ell\rvert>\tilde{N}}c_{\ell}\left(\emph{\textbf{e}}(\ell\tilde{t}_k) - (-1)^{\lfloor\frac{\ell+\tilde{N}}{N}\rfloor}\emph{\textbf{e}}(r(\ell)\tilde{t}_k)\right),
\end{equation}
where $r(\ell) = \rem(\ell + \tilde{N},N) - \tilde{N}$ with $\rem(p,q)$ giving the remainder after division of $p$ by $q$.
As a consequence, if $\tilde{t}_k\in\Omega$ for all $k\in [m]$ then for any integer $1\leq p< \infty$
\begin{align}
\label{errorS}
&\|\tilde{f}-\mathcal{S}f\|_p \leq 2m^{\frac{1}{p}}\sum_{\lvert\ell\rvert>\tilde{N}}\lvert c_{\ell}\rvert,
\end{align}
and
\begin{align}
&\|\tilde{f}-\mathcal{S}f\|_{\infty} \leq 2\sum_{\lvert\ell\rvert>\tilde{N}}\lvert c_{\ell}\rvert.
\end{align}
\end{theorem}
The proof of this theorem is postponed until Section \ref{errorproof}. Therefore, the error due to $\mathcal{S}$ is proportional to the $1$-norm (or Wiener algebra norm) of the Fourier coefficients of $\textbf{f}$ that correspond to frequencies larger than $\tilde{N} = \frac{N-1}{2}$. In particular notice that if $c_{\ell} = 0$ for all $\ell > \tilde{N}$ we obtain perfect interpolation, as expected from standard results in signal processing (i.e., bandlimited signals consisting of trigonometric polynomials with finite degree $\leq\tilde{N}$). Despite the wide usage of trigonometric interpolation in applications \cite{strohmer1,margolis,stability}, such a result that gives a sharp error bound does not seem to exist in the literature. 

Notice that Theorem \ref{thmS} only holds for $\tilde{\tau}\subset\Omega$ as restricted in Section \ref{SignalModel}. However, the results continues to hold for $\tilde{\tau}$ unrestricted if we sample on the torus as imposed in Section \ref{SignalModel}. Therefore, the error bound will always hold under our setup.

\section{Anti-aliasing via Nonuniform Sampling}
\label{MainResult}

With the definitions and assumptions introduced in Section \ref{Methodology}, our methodology in this chapter will consist of modeling our $m$ nonuniform measurements via $\mathcal{S}$ and approximating the $s$ largest coefficients of $f$ in $\Psi$ (in the representation $f=\Psi g$). This discrete approximation will provide an accurate estimate $\textbf{f}^{\sharp}(x)$ of $\textbf{f}(x)$ for all $x\in\Omega$, achieving precision comparable to that given by the best $\frac{N}{2}$-bandlimited approximation of $\textbf{f}$ while requiring $m\ll N$ samples.

The following is a simplified statement, assuming that $\Psi$ is an orthonormal basis and $m\leq N$. We focus on this cleaner result for ease of exposition, presented as a corollary of the main result in Section \ref{mainresultproof}. The full statement considers the case $m\geq N$ and allows for more general and practical $\Psi$ that allow for reduced sample complexity.

\begin{theorem}
\label{CSsimple}
Let $2\leq s\leq N$ and $m\leq N$, where $m$ is the number of nonuniform samples. Under our signal model with Fourier expansion (\ref{Fexpansion}), let $\Psi\in\mathbb{C}^{N\times N}$ be an orthonormal basis with DFT-incoherence parameter $\gamma$. Define the interpolation kernel $\mathcal{S}$ as in Section \ref{dirichletsection} with the entries of $\Delta$ i.i.d. from any distribution satisfying our deviation model from Section \ref{DevModel} with $\theta < 1$. 

Define
\begin{equation}
\label{l1minsimple}
g^{\sharp}\coloneqq\argmin_{h\in\mathbb{C}^N}\lambda\|h\|_1+\frac{\sqrt{N}}{\sqrt{m}}\|\mathcal{S}\Psi h - b\|_2
\end{equation}
with
\[
0<\lambda \leq \frac{\sqrt{1-\theta}}{2\sqrt{2s}}.
\]

If
\begin{equation}
\label{samplecomplexityCS1}
m \geq \frac{C_1\gamma^2(1+\theta)}{(1-\theta)^2}
s\log^4\left(C_2N\right)
\end{equation}
where $C_1$ and $C_2$ are absolute constants, then
\begin{equation}
\label{errorbdsimple}
\|f-\Psi g^{\sharp}\|_{2} \leq \frac{8\epsilon_s(g)}{\sqrt{s}} + \left(\frac{4}{\lambda\sqrt{s}}+\frac{8\sqrt{2}}{\sqrt{1-\theta}}\right)\left(\frac{\sqrt{N}}{\sqrt{m}}\|d\|_2 + 2\sqrt{N}\sum_{\lvert\ell\rvert>\frac{N-1}{2}}\lvert c_{\ell}\rvert\right)
\end{equation}
with probability exceeding $1-\frac{1}{N}$.
\end{theorem}

Therefore, with $m\sim s\log^4(N)$ randomly perturbed samples we can recover $f$ with error (\ref{errorbdsimple}) proportional to the sparse model mismatch $\epsilon_s(g)$, the noise level $\|d\|_2$, and the error of the best $\frac{N-1}{2}$-bandlimited approximation of $\textbf{f}$ in the Wiener algebra norm (i.e., $\sum_{\lvert\ell\rvert>\frac{N-1}{2}}\lvert c_{\ell}\rvert$). As a consequence, we can approximate $\textbf{f}(x)$ for all $x\in\Omega$ as stated in the following corollary.

\begin{corollary}
\label{fullsignalerror}
Let $\emph{\textbf{h}}:\Omega\to\mathbb{C}^N$ be the vector valued function defined entry-wise for $\ell\in[N]$ as
\begin{equation}
\label{hvf}
\emph{\textbf{h}}(x)_{\ell} \coloneqq \frac{1}{\sqrt{N}}\textbf{e}(-x(\ell-\tilde{N}-1)),
\end{equation}
and define the function $\emph{\textbf{f}}^{\sharp}:\Omega\to\mathbb{C}$ via
\begin{equation}
\label{bandapprox}
\emph{\textbf{f}}^{\sharp}(x) = \langle \emph{\textbf{h}}(x),\mathcal{F}^*\Psi g^{\sharp}\rangle,
\end{equation}
where $g^{\sharp}$ is given by (\ref{l1minsimple}).

Then, under the assumptions of Theorem \ref{CSsimple}, 
\begin{align}
\label{errorbdsimple2}
\lvert\emph{\textbf{f}}(x) - \emph{\textbf{f}}^{\sharp}(x)\rvert &\leq \frac{8\epsilon_s(g)}{\sqrt{s}} + \left(\frac{4}{\lambda\sqrt{s}}+\frac{8\sqrt{2}}{\sqrt{1-\theta}}\right)\frac{\sqrt{N}}{\sqrt{m}}\|d\|_2\nonumber\\
&+\left(\frac{8\sqrt{N}}{\lambda\sqrt{s}}+\frac{16\sqrt{2N}}{\sqrt{1-\theta}}+2\right) \sum_{\lvert\ell\rvert>\frac{N-1}{2}}\lvert c_{\ell}\rvert
\end{align}
holds for all $x\in\Omega = [-\frac{1}{2},\frac{1}{2})$ with probability exceeding $1-\frac{1}{N}$.
\end{corollary}
The proof of this corollary is presented in Section \ref{errorproof}. In the case $\epsilon_s(g) = \|d\|_2 = 0$, the results intuitively say that we can recover a $\frac{N-1}{2}$-bandlimited approximation of $\textbf{f}$ with $\mathcal{O}(s\polylog(N))$ random off-the-grid samples. In the case of equispaced samples, $\mathcal{O}(N)$ measurements are needed for the same quality of reconstruction by the Nyquist-Shannon sampling theorem (or by Theorem \ref{thmS} directly). Thus, for compressible signals with $s\ll N$, random nonuniform samples provide a significant reduction in sampling complexity (undersampling) and simultaneously allow recovery of frequency components exceeding the sampling density (anti-aliasing). See Section \ref{discussion} for further discussion.

Notice that general denoising is not guaranteed in an undersampling scenario, due to the term $\frac{\sqrt{N}}{\sqrt{m}}\|d\|_2$ in (\ref{errorbdsimple}), and (\ref{errorbdsimple2}). In other words, one cannot expect the output estimate to reduce the measurement noise $\|d\|_2$ since $\frac{\sqrt{N}}{\sqrt{m}}\geq 1$ appearing in our error bound implies an amplification of the input noise level. Such situations with limited samples are typical in compressive sensing and this noise amplifying behavior is demonstrated numerically in Section \ref{noiseexp}. In general a practitioner must oversample (i.e., $N< m$) to attenuate the effects of generic noise. However, Theorem \ref{CSsimple} and Corollary \ref{fullsignalerror} state that nonuniform samples specifically attenuate aliasing noise.

\section{Denoising via Nonuniform Oversampling}
\label{1D}

In this section, we show that reduction in the noise level introduced during acquisition is possible given nonuniform samples whose average density exceeds the Nyquist rate (relative to a desired bandwidth). While the implications of this section are not surprising in the context of classical sampling theory, to the best of our knowledge such guarantees do not exist in the literature when the sampling points are nonuniform.

By removing the sparse signal model (Section \ref{SignalModel}), deviation model (Section \ref{DevModel}), and requiring $m\geq N$ off-the-grid samples (on the torus, see Section \ref{SignalModel}), we now use the numerically cheaper program of least squares. To reiterate, $\textbf{f}\in A(\Omega)$ with Fourier expansion $\sum_{\ell=-\infty}^{\infty}c_{\ell}\textbf{e}(\ell x)$ is our continuous signal of interest. With $N$ odd, $f \in \mathbb{C}^{N}$ is the discrete signal to be approximated, where $f_k = \textbf{f}(t_{k})$ for $t_k := \frac{k-1}{N}-\frac{1}{2}$. The vector $\tilde{f} \in \mathbb{C}^{m}$ encapsulates the nonuniformly sampled data where $\tilde{f}_k = \textbf{f}(\tilde{t}_{k})$ for $\tilde{t}_k := \frac{k-1}{m}-\frac{1}{2} + \Delta_k$. Noisy nonuniform samples are given as
\[
b = \tilde{f} + d \in \mathbb{C}^m,
\]
where the additive noise model, $d$, does not incorporate off-the-grid corruption.

In this oversampling context, we provide a denoising result for a more general set of deviations.

\begin{theorem}
\label{OSthm}
Let the entries of $\Delta$ be i.i.d. from any distribution and define 
\begin{equation}
\label{ls}
f^{\sharp} \coloneqq \argmin_{h\in\mathbb{C}^N}\|\mathcal{S} h - b\|_2.
\end{equation}
If $m= \kappa N\log(N)$ with $\kappa \geq \frac{4}{\log\left(e/2\right)}$, then
\begin{equation}
\label{boundls}
\|f- f^{\sharp}\|_{2} \leq \frac{2\sqrt{2}}{\sqrt{\kappa\log(N)}}\|d\|_2 + 4\sqrt{2N}\sum_{\lvert\ell\rvert>\frac{N-1}{2}}\lvert c_{\ell}\rvert
\end{equation}
with probability exceeding $1-\frac{1}{N}$.

\end{theorem}
The proof can be found in Section \ref{proofOS}. In this scenario, we oversample relative to the $\frac{N-1}{2}$-bandlimited output by generating a set samples with average density exceeding the Nyquist rate (step size $\frac{1}{N}$). With $m\geq \kappa N\log(N)$ for $\kappa\geq 1$, bound (\ref{boundls}) tells us that we can diminish the measurement noise level $\|d\|_2$ by a factor $ \sim\frac{1}{\sqrt{\kappa\log(N)}}$. The oversampling parameter $\kappa$ may be varied for increased attenuation at the cost of denser sampling. We comment that the methodology from Section \ref{MainResult} with $m\geq N$ also allows for denoising and similar error bounds (see Theorem \ref{CS}). However, focusing on the oversampling case distinctly provides simplified results with many additional benefits.

In particular, here the deviations $\Delta$ need not be from our deviation model in Section \ref{DevModel} and instead the result applies to perturbations generated by any distribution. This includes the degenerate distribution (deterministic), so the claim also holds in the case of equispaced samples. Furthermore, we no longer require the sparsity assumption and the result applies to all functions in the Wiener algebra. Finally, the recovery method (\ref{ls}) consists of standard least squares which can be solved cheaply relative to the square-root LASSO decoder (\ref{l1minsimple}) from the previous section.

We may proceed analogously to Corollary \ref{fullsignalerror} and show that the output discrete signal $f^{\sharp}$ provides a continuous approximation 
\[
\textbf{f}^{\sharp}(x) \coloneqq \langle \textbf{h}(x),\mathcal{F}^*f^{\sharp}\rangle\approx \textbf{f}(x) 
\]
for all $x\in\Omega$, where \textbf{h}(x) is defined in (\ref{hvf}). The error of this estimate is bounded as
\[
\lvert \textbf{f}^{\sharp}(x)-\textbf{f}(x)\rvert \leq \frac{2\sqrt{2}}{\sqrt{\kappa\log(N)}}\|d\|_2 + \left(4\sqrt{2N}+2\right)\sum_{\lvert\ell\rvert>\frac{N-1}{2}}\lvert c_{\ell}\rvert,
\]
proportional to the error of the best $\frac{N-1}{2}$-bandlimited approximation in the Wiener algebra norm while attenuating the introduced measurement noise. In the result, the structure of the deviated samples is quite general and accounts for many practical cases. 

While related results exist in the equispaced case (see for example Section 4 of \cite{oversampling}), Theorem \ref{OSthm} is the first such statement in a general non-equispaced case. The result therefore provides insight into widely applied techniques for the removal of unwanted noise, without making any assumptions on the noise structure.

\section{Discussion}
\label{discussion}

This section elaborates on several aspects of the results. Section \ref{novelty} discusses relevant work in the literature. Section \ref{DeviationDiscussion} provides examples of distributions that satisfy our deviation model and intuition of its meaning. Section \ref{SignalModelDiscussion} explores the $\gamma$ parameter with examples of transformations $\Psi$ that produce a satisfiable sampling complexity.

\subsection{Related Work}
\label{novelty}

Several studies in the compressive sensing literature are similar to our results in Section \ref{MainResult} (\cite{structured,stability}). In contrast to these references, we derive recovery guarantees for non-orthonormal systems (when $\theta\neq 0$) while focusing the scope of the paper within the context of classical sampling theory (introducing error according to the bandlimited approximation). The work in \cite{stability} considers sampling of sparse trigonometric polynomials and overlaps with our application in the case $\Psi = \mathcal{F}$. Our results generalize this work to allow for other signal models and sparsifying transforms. Furthermore, \cite{stability} assumes that the samples are chosen uniformly at random from a continuous interval or a discrete set of $N$ equispaced points. In contrast, our results pertain to general deviations from an equispaced grid with average sampling density $\sim s\polylog(N)$ and allow for many other distributions of the perturbations.

\subsection{Deviation Model}
\label{DeviationDiscussion}
In this section, we present several examples of distributions that are suitable for our results in Section \ref{MainResult}. Notice that our deviation model utilizes the characteristic function of a given distribution, evaluated at a finite set of points. This allows one to easily consider many distributions for our purpose by consulting the relevant and exhaustive literature of characteristic functions (see for example \cite{characteristic}).
\begin{enumerate}
	\item \underline{Uniform continuous}: $\mathcal{D} = \mathcal{U}[-\frac{1}{2m},\frac{1}{2m}]$ gives $\theta = 0$. To generalize this example, we may take $\mathcal{D} = \mathcal{U}[\mu-\frac{p}{2m},\mu+\frac{p}{2m}]$, for any $\mu\in\mathbb{R}$ and $p\in\mathbb{N}/\{0\}$ to obtain $\theta = 0$ (i.e., shift and dilate on the torus). Notice that with $p=m$, we obtain i.i.d. samples chosen uniformly from the whole interval $\Omega$ (as in \cite{stability}).
  \item \underline{Uniform discrete}: $\mathcal{D} = \mathcal{U}\{- \frac{1}{2m} + \frac{k}{m\bar{n}}\}_{k=0}^{\bar{n}-1}$ with $\bar{n}:=\lceil\frac{2(N-1)}{m}\rceil + 1$ gives $\theta = 0$. To generalize, we may shift and dilate $\mathcal{D} = \mathcal{U}\{\mu - \frac{p}{2m} + \frac{pk}{m\bar{n}}\}_{k=0}^{\bar{n}-1}$, for any $\mu\in\mathbb{R}$, $p\in\mathbb{N}/\{0\}$ and integer $\bar{n}>\frac{2(N-1)p}{m}$. We obtain $\theta = 0$ as well. 
	\item \underline{Normal}: $\mathcal{D} = \mathcal{N}(\mu,\bar{\sigma}^2)$, for any mean $\mu\in\mathbb{R}$ and variance $\bar{\sigma}^2>0$. Here 
 \[
 \theta = \frac{2N}{m}e^{-2(\bar{\sigma}\pi m)^2}. 
 \]
 In particular, for fixed $\bar{\sigma}$, $m$ may be chosen large enough to satisfy the conditions of Theorem \ref{CS} and vice versa.
	\item \underline{Laplace}: $\mathcal{D} = \mathcal{L}(\mu,b)$, for any location $\mu\in\mathbb{R}$ and scale $b>0$ gives 
 \[
 \theta = \frac{2N}{m(1+(2\pi bm)^2)}.
 \]
	\item \underline{Exponential}: $\mathcal{D} = \mbox{Exp}(\lambda)$, for any rate $\lambda>0$ gives 
 \[
 \theta = \frac{2N}{m\sqrt{1+4\pi^2m^2\lambda^{-2}}}.
 \]
\end{enumerate}
In particular, notice that examples 1 and 2 include cases of jittered sampling \cite{non15,non16,non17,non14}. Indeed, with $p=1$ these examples partition $\Omega$ into $m$ regions of equal size and these distributions will choose a point randomly from each region (in a continuous or discrete sense). The jittered sampling list can be expanded by considering other distributions to generate samples within each region. 

In general we will have $\theta > 0$, which implies deteriorated output quality and increases the number of required off-the-grid samples according to Theorem \ref{CSsimple}. Arguably, our deviation model introduces a notion of \emph{optimal} jitter when the chosen distribution achieves $\theta = 0$, ideal in our results. This observation may be of interest in the active literature of jittered sampling techniques \cite{non14}.

Intuitively, $\theta$ is measuring how biased a given distribution is in generating deviations. If $\delta\sim\mathcal{D}$, $\lvert\mathbb{E}\textbf{e}(jm\delta)\rvert \approx 0$ means that the distribution is nearly centered and impartial. On the other hand, $\lvert\mathbb{E}\textbf{e}(jm\delta)\rvert \approx 1$ gives the opposite interpretation where the deviations will be generated favoring a certain direction in an almost deterministic sense. Our result is not applicable to such biased distributions, since in Theorem \ref{CSsimple} as $\theta \rightarrow 1$ the error bound becomes unbounded and meaningless.

\subsection{Signal Model}
\label{SignalModelDiscussion}

In this section we discuss the DFT-incoherence parameter $\gamma$, introduced in Section \ref{SignalModel} as 
\[
\gamma = \max_{\ell\in [n]} \sum_{k=1}^{N}\lvert\langle \mathcal{F}_{*k},\Psi_{*\ell}\rangle\rvert,
\]
where we now let $\Psi\in\mathbb{C}^{N\times n}$ be a full column rank matrix with $n\leq N$. The parameter $\gamma$ a uniform upper bound on the $1$-norm of the discrete Fourier coefficients of the columns of $\Psi$. Since the decay of the Fourier coefficients of a function is related to its smoothness, intuitively $\gamma$ can be seen as a measure of the smoothness of the columns of $\Psi$. Implicitly, this also measures the smoothness of $\textbf{f}$ since its uniform discretization admits a representation via this transformation $f = \Psi g$. 

Therefore, the role of $\gamma$ on the sampling complexity is clear, relatively small $\gamma$ implies that our signal of interest is smooth and therefore requires less samples. This observation is intuitive, since non-smooth functions will require additional samples to capture discontinuities in accordance with Gibbs phenomenon. This argument is validated numerically in Section \ref{dftexp}, where we compare reconstruction via an infinitely differentiable ensemble (FFT) and a discontinuous wavelet (Daubechies 2).

We now consider several common choices for $\Psi$ and discuss the respective $\gamma$ parameter:
\begin{enumerate}
  \item $\Psi = \mathcal{F}$ (the DFT), then $\gamma = 1$ which is optimal. However, most appropriate and common is the choice $\Psi = \mathcal{F}^*$ which can be shown to exhibit $\gamma\sim \mathcal{O}(1)$ by a simple calculation.
  \item When $\Psi = \mathcal{H}^*$ is the inverse 1D Haar wavelet transform, we have $\gamma\sim \mathcal{O}(\log(N))$. In \cite{haar} it is shown that the inner products between rows of $\mathcal{F}$ and rows of $\mathcal{H}$ decay according to an inverse power law of the frequency (see Lemma 1 therein). A similar proof shows that $\lvert\langle \mathcal{F}_{*k},\mathcal{H}^{*}_{*\ell}\rangle\rvert\sim \frac{1}{\lvert k\rvert}$, which gives the desired upper bound for $\gamma$ via an integral comparison. Notice that these basis vectors have jump discontinuities, and yet we still obtain an acceptable DFT-incoherence parameter for nonuniform undersampling.
	\item $\Psi = \mathcal{I}_N$ (the $N\times N$ identity) gives $\gamma = \sqrt{N}$. This is the worst case scenario for normalized transforms since
	\[
\max_{v\in S^{N-1}} \sum_{k=1}^{N}\lvert\langle \mathcal{F}_{*k},v\rangle\rvert = \max_{v\in S^{N-1}} \sum_{k=1}^{N}\lvert\langle \mathcal{F}_{*k},\mathcal{F}^*v\rangle\rvert = \max_{v\in S^{N-1}} \sum_{k=1}^{N}\lvert v_k\rvert = \sqrt{N}.
	\]
	In general, our smooth signals of interest are not fit for this sparsity model.
	\item Let $p\geq 1$ be an integer. Matrices $\Psi$ whose columns are uniform discretizations of $p$-differentiable functions, with $p-1$ periodic and continuous derivatives and $p$-th derivative that is piecewise continuous. In this case $\gamma \sim \mathcal{O}(\log(N))$ if $p=1$ and $\gamma \sim \mathcal{O}(1)$ if $p\geq 2$. For sake of brevity we do not provide this calculation, but refer the reader to Section 2.8 in \cite{thesis} for an informal argument.
\end{enumerate}
Example 4 is particularly informative due to its generality and ability to somewhat formalize the intuition behind $\gamma$ previously discussed. This example implies the applicability of our result to a general class of smooth functions that agree nicely with our signal model defined in Section \ref{SignalModel} (functions in $A(\Omega)$).

\section{Numerical Experiments}
\label{experiments}

In this section we present numerical experiments to explore several aspects of our methodology and results. Specifically, we consider the effects of the DFT-incoherence and $\theta$ parameter in Section \ref{dftexp} and \ref{thetaexp} respectively. Section \ref{noiseexp} investigates the noise attenuation properties of nonuniform samples. We first introduce several terms and models to describe the setup of the experiments. Throughout we let $N = 2015$ be the size of the uniformly discretized signal $f$.

Program (\ref{l1minS}) with $\lambda = \frac{1}{2\sqrt{2s}}$ is solved using CVX \cite{cvx1,cvx2}, a MATLAB\textsuperscript{\tiny\textregistered} optimization
toolbox for solving convex problems. We implement the Dirichlet kernel using (\ref{dirichlet}) directly to construct $\mathcal{S}$. We warn the reader that in this section we have not dedicated much effort to optimize the numerical complexity of the interpolation kernel. For a faster implementation, we recommend instead applying the DFT/NDFT representation $\mathcal{S}= \mathcal{N}\mathcal{F}^*$ (see Section \ref{dirichletsection}) using NFFT 3 software from \cite{NFFT} or its parallel counterpart \cite{NFFTpara}. 

Given output $f^{\sharp} = \Psi g^{\sharp}$ with true solution $f$, we consider the relative norm of the reconstruction error as a measure of output quality, given as
\[
\mbox{Relative Error} = \frac{\|f^{\sharp}-f\|_2}{\|f\|_2}.\\
\]

\emph{Grid perturbations:} To construct the nonuniform grid $\tilde{\tau}$, we introduce an irregularity parameter $\rho\geq 0$. We define our perturbations by sampling from a uniform distribution, so that each $\Delta_k$ is drawn uniformly at random from $[-\frac{\rho}{m},\frac{\rho}{m}]$ for all $k\in[m]$ independently. Off-the-grid samples $\tilde{\tau}$ are generated independently for each signal reconstruction experiment.

\emph{Complex exponential signal model:} We consider bandlimited complex exponentials with random harmonic frequencies. With bandwidth $\omega = \frac{N-1}{2} = 1007$, and sparsity level $s = 50$ we generate $\vec{\omega}\in\mathbb{Z}^s$ by choosing $s$ frequencies uniformly at random from $\{-\omega,-\omega+1,\cdots,\omega\}$ and let
\[
\textbf{f}(x) = \sum_{k=1}^{s}\textbf{e}\left(\vec{\omega}_k x\right).
\]
We use the DFT as a sparsifying transform $\Psi = \mathcal{F}$ so that $g = \Psi^{*} f = \Psi^{-1} f$ is a $50$-sparse vector. This transform is implemented using MATLAB's fft function. The frequency vector, $\vec{\omega}$, is generated randomly for each independent set of experiments. Note that in this case we have optimal DFT-incoherence parameter $\gamma = 1$ (see Section \ref{SignalModelDiscussion}).

\emph{Gaussian signal model:} We consider a non-bandlimited signal consisting of sums of Gaussian functions. This signal model is defined as
\[
\textbf{f}(x) = -e^{-100x^2}+e^{-100(x-.104)^2}-e^{-100(x+.217)^2}.
\]
For this dataset, we use the Daubechies 2 wavelet as a sparsifying transform $\Psi$, implemented using the Rice Wavelet Toolbox \cite{rice}. This provides $g = \Psi^{*} f = \Psi^{-1} f$ that can be well approximated by a $50$-sparse vector. In other words, all entries of $g$ are non-zero but $\epsilon_{50}(g) < .088 \approx \frac{\|f\|_2}{250}$ and if $g_{50}$ is the best 50-sparse approximation of $g$ then $\|f-\Psi g_{50}\|_2 < .026 \approx \frac{\|f\|_2}{850}$. The smallest singular value of the transform is $\sigma_{2015}(\Psi) = 1$ and we have $\gamma \approx 36.62$, computed numerically.

\subsection{Effect of DFT-incoherence}
\label{dftexp}

This section is dedicated to exploring the effect of the DFT-incoherence parameter in signal reconstruction. We consider the complex exponential and Gaussian signal models described above. Recall that in the complex exponential model we have $\Psi = \mathcal{F}$ (the DFT) with optimal DFT-incoherence parameter $\gamma = 1$. In the Gaussian model $\Psi$ is the Daubechies 2 wavelet with $\gamma \approx 36.62$. Varying the number of nonuniform samples, we will compare the quality of reconstruction using both signal models with respective transforms to investigate the role of $\gamma$ in the reconstruction error. We consider the sparsity level $s=50$ and solve (\ref{l1minS}) with $\lambda = \frac{1}{2\sqrt{2s}} = \frac{1}{20}$, though the Gaussian signal model is not 50-sparse in the Daubechies domain (see last paragraph of this subsection for further discussion). 

Here we set irregularity parameter $\rho = \frac{1}{2}$ to generate the deviations (so that $\theta = 0$) and vary the average step size of the nonuniform samples. We do so by letting $m$ vary through the set $\{\lfloor\frac{N}{1.5}\rfloor,\lfloor\frac{N}{2}\rfloor,\lfloor\frac{N}{2.5}\rfloor,\cdots,\lfloor\frac{N}{10}\rfloor\}$. For each fixed value of $m$, the average relative error is obtained by averaging the relative errors of 50 independent reconstruction experiments. The results are shown in Figure \ref{fig:FvsD}, where we plot the average step size vs average relative reconstruction error.

\begin{figure}[htp]
\centering
\includegraphics[width=.75\textwidth]{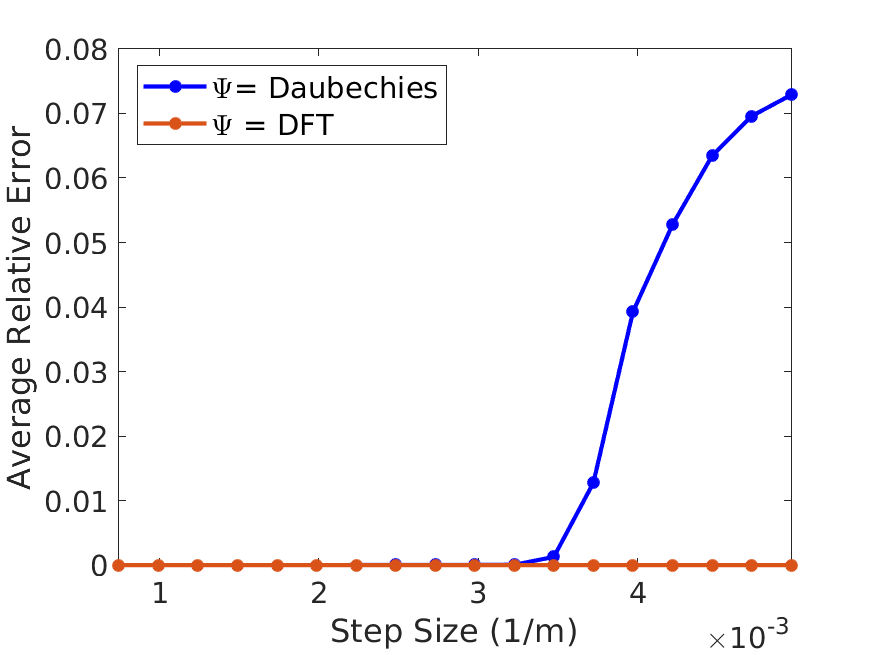}
\caption{Plot of average relative reconstruction error vs average step size for both signal models. In the complex exponential model ($\Psi =\mathcal{F}$, the DFT) we have $\gamma = 1$ and in the Gaussian signal model we have $\gamma \approx 36.62$ (Daubechies 2 wavelet). Notice that the complex exponential model allows for reconstruction from larger step sizes in comparison to the Gaussian signal model.}
\label{fig:FvsD}
\end{figure}

These experiments demonstrate the negative effect of larger DFT-incoherence parameters in signal reconstruction. Indeed, in Figure \ref{fig:FvsD} we see that the complex exponential model with $\gamma =1$ allows for accurate reconstruction from larger step sizes. This is to be expected from Section \ref{MainResult}, where the results imply that the Daubechies 2 wavelet will require more samples for successful reconstruction according to its parameter $\gamma \approx$ 36.62.

To appropriately interpret these experiments, it is important to note that the signal from the Gaussian model is only compressible and does not exhibit a 50-sparse representation via the Daubechies transform. Arguably, this may render the experiments of this section inappropriate to purely determine the effect of $\gamma$ since the impact of approximating the Gaussian signal with a 50-sparse vector may be significant and produce an unfair comparison (i.e., due the sparse model mismatch term $\epsilon_{50}(g)$ appearing in our error bound (\ref{errorbdsimple})). This is important for the reader to keep in mind, but we argue that the effect of this mismatch is negligible since in the Gaussian signal model with $g = \Psi^{-1} f$ we have $\epsilon_{50}(g) < \frac{\|f\|_2}{250}$ and if $g_{50}$ is the best 50-sparse approximation of $g$ then $\|f-\Psi g_{50}\|_2 < \frac{\|f\|_2}{850}$. This argument can be further validated with modified numerical experiments where $f$ does have a 50-sparse representation in the Daubechies domain, producing reconstruction errors with identical behavior and magnitude as those in Figure \ref{fig:FvsD}. Therefore, we believe our results here are informative to understand the impact of $\gamma$. For brevity, we do not present these modified experiments since such an $f$ will not longer satisfy the Gaussian signal model and complicate our discussion.

\subsection{Effect of the Deviation Model Parameter}
\label{thetaexp}

In this section we generate the deviations in such a way that vary the deviation model parameter $\theta$, in order to explore its effect on signal reconstruction. We only consider the complex exponential signal model for this purpose and fix $m = \lfloor\frac{N}{10}\rfloor = 201$.

We vary $\theta$ by generating deviations with irregularity parameter $\rho$ varying over $\{.001,.002,.003,\cdots,.009\}\bigcup\{.01,.02,.03,\cdots,.5\}$. For each fixed $\rho$ value we compute the average relative reconstruction error of 50 independent experiments. Notice that for each $k\in[m]$ and any $j$
\[
\mathbb{E}\textbf{e}\left(jm\Delta_k\right) = \frac{\sin\left(2\pi j \rho\right)}{2\pi j \rho}.
\]
Given $\rho$, we use this observation and definition (\ref{thetadef}) to compute the respective $\theta$ value by considering the maximum of the expression above over all $0<\lvert j\rvert\leq \lfloor\frac{N-1}{m}\rfloor = 10$. The relationship between $\rho$ and $\theta$ is illustrated in Figure \ref{fig:rho} (right plot), where smaller irregularity parameters $\rho\approx 0$ provide larger deviation model parameters $\theta$.

According to (\ref{thetadef}), this allows $\theta\in [0,20.05)$, which violates the assumption $\theta < 1$ of Theorem \ref{CSsimple} and does not allow (\ref{l1minS}) to be implemented with parameter in the required range
\[
0<\lambda \leq \frac{\sqrt{1-\theta}}{2\sqrt{2s}}.
\]
Despite this, we implement all experiments in this section with $\lambda = \frac{1}{2\sqrt{2s}} = \frac{1}{20}$ (where $s=50$). Such a fixed choice may not provide a fair set of results, since the parameter is not adapted in any way to the deviation model. Regardless, the experiments will prove to be informative while revealing the robustness of the square-root LASSO decoder with respect to parameter selection.

Figure \ref{fig:rho} plots $\theta$ vs average relative reconstruction error (left plot). In the plot, our main result (Theorem \ref{CSsimple}) is only strictly applicable in three cases (outlined in red, $\theta = 0,.409,.833$). However, the experiments demonstrate that decent signal reconstruction may be achieved when the condition $\theta < 1$ does not hold and the parameter $\lambda$ is not chosen appropriately. Therefore, the applicability of the methodology goes beyond the restrictions of the theorem and the numerical results demonstrate the flexibility of the square-root LASSO decoder.

\begin{figure}[htp]
\centering
\includegraphics[width=.495\textwidth]{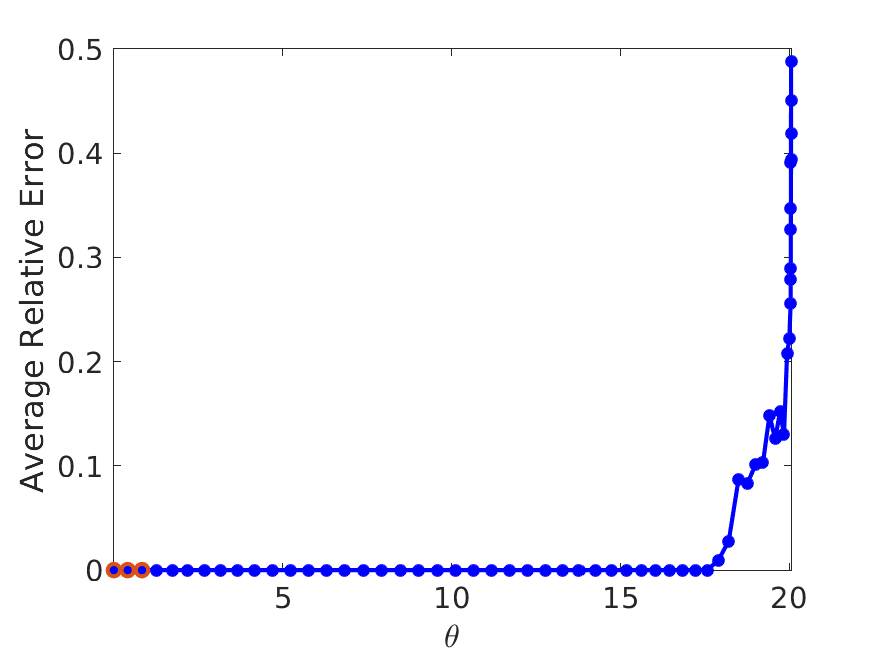}
\includegraphics[width=.495\textwidth]{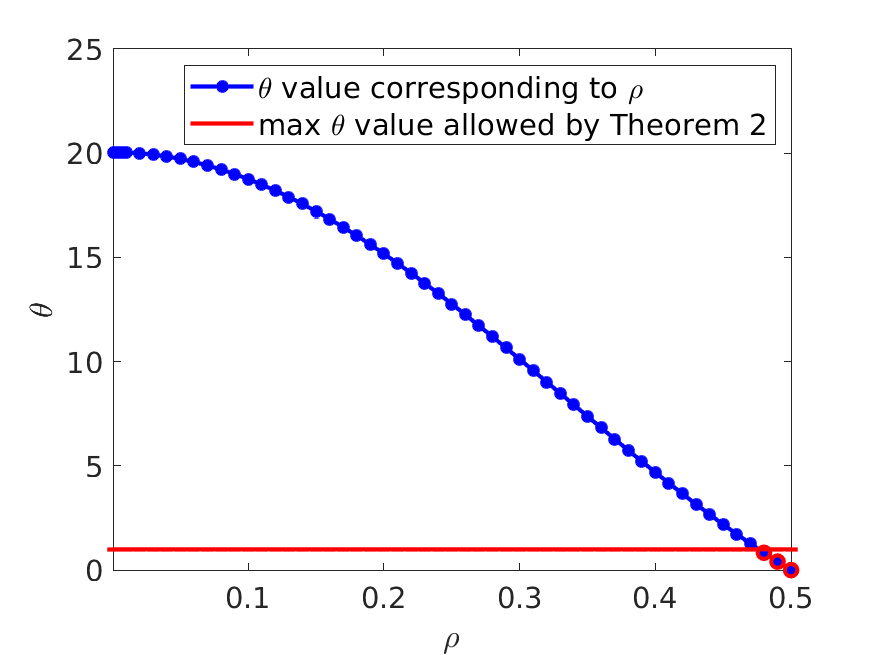}
\caption{(left) Plot of average relative reconstruction error vs corresponding $\theta$ parameter and (right) plot illustrating the relationship between the irregularity parameter $\rho$ and the deviation model parameter $\theta$. The plots emphasize via red outlines the $\theta$ values that satisfy the conditions of Theorem \ref{CSsimple} (i.e., $\theta < 1$). Although our results only hold for three $\theta$ values ($0,.409,.833$), the experiments demonstrate that accurate recovery is possible otherwise.}
\label{fig:rho}
\end{figure}

\subsection{Noise Attenuation}
\label{noiseexp}

This section explores the robustness of the methodology when presented with measurement noise, in both the undersampled and oversampled cases relative to the target bandwidth $\frac{N-1}{2}$ (Sections \ref{MainResult} and \ref{1D} respectively). We only solve the square-root LASSO problem (\ref{l1minS}) with $\lambda = \frac{1}{2\sqrt{2s}} = \frac{1}{20}$, and avoid the least squares problem (\ref{ls}) for brevity. However, we note that both programs produce similar results and conclusions in the oversampled case (see Theorem \ref{CS}). We only consider the bandlimited complex exponential signal model for this purpose. We generate additive random noise $d\in\mathbb{R}^m$ from a uniform distribution. Each entry of $d\in\mathbb{R}^{m}$ is i.i.d. from $[-\frac{\chi}{1000},\frac{\chi}{1000}]$ where $\chi = \frac{1}{\sqrt{m}}\|f\|_1$, chosen to maintain $\|d\|_2$ relatively constant as $m$ varies.

We set $\rho = \frac{1}{2}$ to generate the deviations (so that $\theta = 0$) and vary the average step size of the nonuniform samples. We do so by letting $m$ vary through the set $\{\lfloor\frac{N}{.5}\rfloor,\lfloor\frac{N}{.75}\rfloor,N,\cdots,\lfloor\frac{N}{6.75}\rfloor,\lfloor\frac{N}{7}\rfloor\}$, notice that only the first two cases correspond to oversampling. For each fixed value of $m$, the relative reconstruction error is obtained by averaging the result of 50 independent experiments. The results are shown in Figure \ref{fig:noise}, where we plot the average step size vs average relative reconstruction error and average relative input noise level $\|d\|_2/\|f\|_2$.

\begin{figure}[htp]
\centering
\includegraphics[width=.75\textwidth]{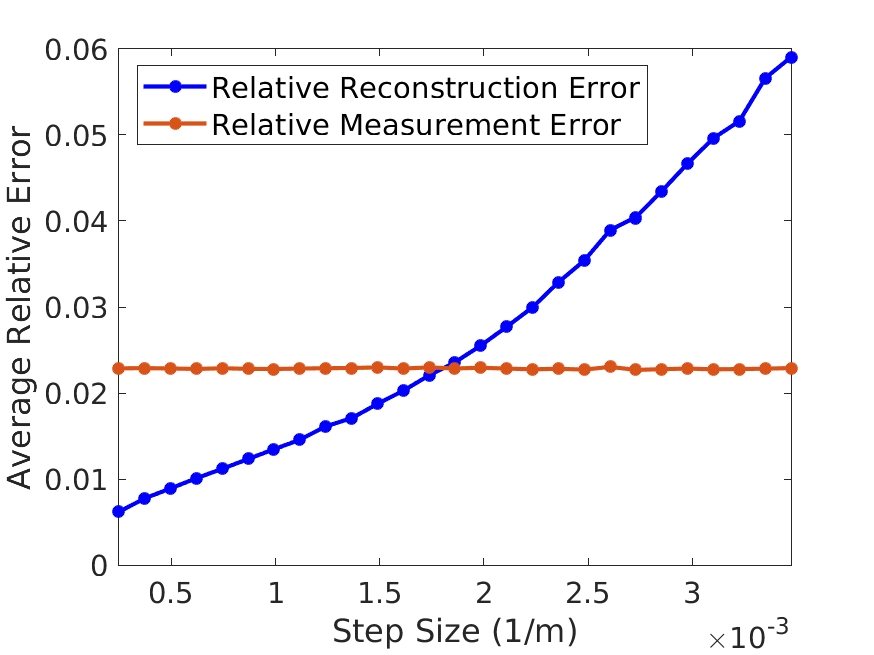}
\caption{Plot of average relative reconstruction error ($\|f-f^{\sharp}\|_2/\|f\|_2$) vs average step size (blue curve) and average input relative measurement error ($\|d\|_2/\|f\|_2$) vs average step size (red curve). Notice that the first 13 step size values achieve noise attenuation, i.e., the reconstruction error is lower than the input noise level.}
\label{fig:noise}
\end{figure}

The first two cases ($m = \lfloor\frac{N}{.5}\rfloor,\lfloor\frac{N}{.75}\rfloor$) correspond to oversampling and illustrate the results from Section \ref{1D} (and Theorem \ref{CS}), where attenuation of the input noise level is achieved. Surprisingly, these experiments demonstrate that nonuniform undersampling also allows for denoising. This is seen in Figure \ref{fig:noise}, where the values $m = \lfloor\frac{N}{1.25}\rfloor,\lfloor\frac{N}{1.5}\rfloor,\cdots,\lfloor\frac{N}{3.5}\rfloor$ correspond to sub-Nyquist rates and output an average relative reconstruction error smaller than the input measurement noise level. Thus, when nonuniform samples are not severely undersampled, the negative effects of random noise can be reduced.

\section{Conclusions}
\label{conclusion}

This paper provides a concrete framework to study the benefits of random nonuniform samples for signal acquisition (in comparison to equispaced sampling), with explicit statements that are informative for practitioners. Related observations are extensive but largely empirical in the sampling theory literature. Therefore, this work supplies novel theoretical insights on this widely discussed phenomenon. In the context of compressive sensing, we extend the applicability of this acquisition paradigm by demonstrating how it naturally intersects with standard sampling techniques. We hope that these observations will prompt a broader usage of compressive sensing in real world applications that rely on classical sampling theory.

There are several avenues for future research. First, the overall methodology requires the practitioner to know the nonuniform sampling locations $\tilde{\tau}$ accurately. While this is typical for signal reconstruction techniques that involve non-equispaced samples, it would be of practical interest to extend the methodology is such a way that allows for robustness to inaccurate sampling locations and even self-calibration. Further, as mentioned in Section \ref{experiments}, this work has not dedicated much effort to a numerically efficient implementation of the Dirichlet kernel $\mathcal{S}$. This is crucial for large-scale applications, where a direct implementation of the Dirichlet kernel via its Fourier or Dirichlet representation (see \cite{offgrid}) may be too inefficient for practical purposes. As future work, it would be useful to consider other interpolation kernels with greater numerical efficiency (e.g., a low order Lagrange interpolation operator). 

Finally, to explore the undersampling and anti-aliasing properties of nonuniform samples, our results here require a sparse signal assumption and adopt compressive sensing methodologies. However, most work that first discussed this nonuniform sampling phenomenon precedes the introduction of compressive sensing and does not explicitly impose sparsity assumptions. Therefore, to fully determine the benefits provided by off-the-grid samples it would be most informative to consider a more general setting, e.g., only relying on the smoothness of continuous-time signals. We believe the work achieved here provides a potential avenue to do so.

\section*{Acknowledgments}

This work was in part financially supported by the Natural Sciences and Engineering Research Council of Canada (NSERC) Collaborative Research and Development Grant DNOISE II (375142-08). This research was carried out as part of the SINBAD II project with support from the following organizations: BG Group, BGP, CGG, Chevron, ConocoPhillips, DownUnder GeoSolutions, Hess Corporation, Petrobras, PGS, Sub Salt Solutions, WesternGeco, and Woodside. \"Ozg\"ur Y\i lmaz also acknowledges an NSERC Discovery Grant (22R82411) and an NSERC Accelerator Award (22R68054). 

\section{Proofs}
\label{mainproof}

We now provide proofs to all of our claims. In Section \ref{mainresultproof} we prove Theorem \ref{CSsimple} via a more general result. Theorem \ref{OSthm} is proven in Section \ref{proofOS}. Section \ref{errorproof} establishes the Dirichlet kernel error bounds in Theorem \ref{thmS} and Corollary \ref{fullsignalerror}.

\subsection{Proof of Theorem \ref{CSsimple}}
\label{mainresultproof}

In this section, we will prove a more general result than Theorem \ref{CSsimple} assuming that $\Psi$ is a full column-rank matrix and allowing $m\geq N$. Theorem \ref{CSsimple} will follow from Theorem \ref{CS} by taking $\alpha = \beta = 1, n=N$, and simplifying some terms.

\begin{theorem}
\label{CS}
Let $2\leq s\leq n\leq N$ and $\Psi\in\mathbb{C}^{N\times n}$ be a full column rank matrix with DFT-incoherence parameter $\gamma$ and extreme singular values $\sigma_1(\Psi)\coloneqq\beta\geq\alpha\coloneqq\sigma_n(\Psi)>0$. Let the entries of $\Delta$ be i.i.d. from any distribution satisfying our deviation model with $\theta<1$. Define
\begin{equation}
\label{l1minfull}
g^{\sharp}\coloneqq\argmin_{h\in\mathbb{C}^n}\lambda\|h\|_1+\frac{\sqrt{N}}{\sqrt{m}}\|\mathcal{S}\Psi h - b\|_2
\end{equation}
with
\[
0<\lambda \leq \frac{\alpha\sqrt{1-\theta}}{2\sqrt{2s}}.
\]

If
\begin{align}
\label{samplecomplexityCS}
&m \geq \frac{C_1\gamma^2\beta^2(1+\theta)}{\alpha^4(1-\theta)^2}s\cdot\nonumber\\
&\left(\log\left(\frac{C_2\gamma^2\beta^2(1+\theta)}{\alpha^4(1-\theta)^2}s + 2\right)\log^2\left(\frac{C_2\beta^2(1+\theta)}{\alpha^2(1-\theta)}s\right)\log(n) + \log(n)\right)
\end{align}
where $C_1$ and $C_2$ are absolute constants, then
\[
\|f-\Psi g^{\sharp}\|_{2} \leq \frac{8\beta\epsilon_s(g)}{\sqrt{s}} + \left(\frac{4\beta}{\lambda\sqrt{s}}+\frac{8\beta\sqrt{2}}{\alpha\sqrt{1-\theta}}\right)\left(\frac{\sqrt{N}}{\sqrt{m}}\|d\|_2 + 2\sqrt{N}\sum_{\lvert\ell\rvert>\frac{N-1}{2}}\lvert c_{\ell}\rvert\right)
\]
with probability exceeding $1-\frac{1}{n}$.
\end{theorem}

This theorem generalizes Theorem \ref{CSsimple} to more general transformations $\Psi$ for sparse representation. This is more practical since the columns of $\Psi$ need not be orthogonal, instead linear independence suffices (with knowledge of the singular values $\alpha,\beta$). In particular notice that (\ref{samplecomplexityCS}) depends on $n$ and does not involve $N$, as opposed to $m\sim s\log^4(N)$ in (\ref{samplecomplexityCS1}). Since $n\leq N$, this general result allows for a potential reduction in sample complexity if the practitioner may construct $\Psi$ in such an efficient manner while still allowing a sparse and accurate representation of $f$.

Furthermore, notice that this more general result allows for oversampling $m\geq n$ or $m\geq N$. If we apply Theorem \ref{CS} with $s=n$ then $\epsilon_s(g) = 0$ and we obtain an error bound similar to those in Section \ref{1D}, reducing additive noise by a factor $\frac{\sqrt{N}}{\sqrt{n}\log^2(n)}$ from $m\sim n\log^4(n)$ off-the-grid samples. However, in this scenario the sparsifying transform is no longer of much relevance and it is arguably best to consider the approach of Section \ref{1D} which removes the need to consider $\gamma, \beta, \alpha$, and $\theta$ via a numerically cheaper methodology and a more general set of deviations.

To establish Theorem \ref{CS} we will consider the $\mathcal{G}$-adjusted restricted isometry property ($\mathcal{G}$-RIP) \cite{ben}, defined as follows:
\begin{definition}[$\mathcal{G}$-adjusted restricted isometry property \cite{ben}]
Let $1\leq s\leq n$ and $\mathcal{G}\in\mathbb{C}^{n\times n}$ be invertible. The $s$-th $\mathcal{G}$-adjusted Restricted Isometry Constant ($\mathcal{G}$-RIC) $\delta_{s,\mathcal{G}}$ of a matrix $\mathcal{A}\in\mathbb{C}^{m\times n}$ is the smallest $\delta > 0$ such that
\[
(1-\delta)\|\mathcal{G}v\|_2^2 \leq \|\mathcal{A}v\| \leq (1+\delta)\|\mathcal{G}v\|_2^2
\]
for all $v\in\{z\in\mathbb{C}^{n} \ \lvert \ \|z\|_0 \leq s\}$. If $0< \delta_{s,\mathcal{G}} < 1$ then the matrix $\mathcal{A}$ is said to satisfy the $\mathcal{G}$-adjusted Restricted Isometry Property ($\mathcal{G}$-RIP) of order $s$.
\end{definition}
This property ensures that a measurement matrix is well conditioned amongst all $s$-sparse signals, allowing for successful compressive sensing from $\sim s\polylog(n)$ measurements. Once established for our measurement ensemble, Theorem \ref{CS} will follow by applying the following result:
\begin{theorem}[Theorem 13.9 in \cite{ben}]
\label{GRIP}
Let $\mathcal{G}\in\mathbb{C}^{n\times n}$ be invertible and $\mathcal{A}\in\mathbb{C}^{m\times n}$ have the  $\mathcal{G}$-RIP of order $q$ and constant $0<\delta <1$ where
\begin{equation}
\label{tsparsity}
q = 2\left \lceil{4s\left(\frac{1+\delta}{1-\delta}\right)\|\mathcal{G}\|^2\|\mathcal{G}^{-1}\|^2}\right\rceil .
\end{equation}
Let $g\in\mathbb{C}^{n}$, $y = \mathcal{A}g+d\in\mathbb{C}^{m}$, and $\lambda\leq \frac{\sqrt{1-\delta}}{2\|\mathcal{G}^{-1}\|\sqrt{s}}$. Then 
\[
g^{\sharp}=\argmin_{h\in\mathbb{C}^n}\lambda\|h\|_1+\|\mathcal{A}h - y\|_2
\]
satisfies 
\begin{equation}
\label{Gerror}
\|g-g^{\sharp}\|_2 \leq \frac{8\epsilon_s(g)}{\sqrt{s}} + 8\left(\frac{1}{2\lambda\sqrt{s}}+\frac{\|\mathcal{G}^{-1}\|}{\sqrt{1-\delta}}\right)\|d\|_2.
\end{equation}
\end{theorem}

We therefore obtain our main result if we establish the $\mathcal{G}$-RIP for 
\[
\mathcal{A}:=\sqrt{N}\mathcal{S}\Psi .
\]
To do so, we note that our measurement ensemble is generated from a nondegenerate collection of independent families of random vectors. Such random matrices have been shown to possess the $\mathcal{G}$-RIP in the literature. To be specific, a nondegenerate collection is defined as follows:
\begin{definition}[Nondegenerate collection \cite{ben}]
Let $\boldsymbol{\mathscr{A}}_1,\cdots, \boldsymbol{\mathscr{A}}_m$ be independent families of random vectors on $\mathbb{C}^{n}$. The collection $\boldsymbol{\mathscr{C}}=\{\boldsymbol{\mathscr{A}}_k\}_{k=1}^{m}$ is nondegenerate if the matrix
\[
\frac{1}{m}\sum_{k=1}^{m}\mathbb{E}\left(a_ka_k^*\right),
\]
where $a_k\sim\boldsymbol{\mathscr{A}}_k$, is positive-definite. In this case, write $\mathcal{G}_{\boldsymbol{\mathscr{C}}}\in\mathbb{C}^{n\times n}$ for its unique positive-definite  square root.

\end{definition}

Our ensemble fits this definition, with the rows of $\mathcal{N}\in\mathbb{C}^{m\times N}$ generated from a collection of $m$ independent families of random vectors:
\begin{align*}
    \mathcal{N}_{k*}^* &= \frac{1}{\sqrt{N}}\begin{bmatrix}
           \textbf{e}\left(-\tilde{N}\left(\frac{k-1}{m} - \frac{1}{2} + \Delta_k\right)\right) \\
           \textbf{e}\left(-(\tilde{N}-1)\left(\frac{k-1}{m} - \frac{1}{2} + \Delta_k\right)\right) \\
           \vdots \\
           \textbf{e}\left(\tilde{N}\left(\frac{k-1}{m} - \frac{1}{2} + \Delta_k\right)\right)
         \end{bmatrix} \ \ \mbox{with} \ \ \Delta_k \sim \mathcal{D}.
  \end{align*}
Therefore, in our scenario, the $k$-th family $\boldsymbol{\mathscr{A}}_k$ independently generates deviation $\Delta_k \sim \mathcal{D}$ and produces a random vector of the form above as the $k$-th row of $\mathcal{N}$. This in turn also generates the rows of $\mathcal{A}$ independently, since its $k$-th row is given as $\sqrt{N}\mathcal{N}_{k*}\mathcal{F}^*\Psi$. To apply $\mathcal{G}$-RIP results from the literature for such matrices, we will have to consider the coherence of our collection:
\begin{definition}[Coherence of an unsaturated collection $\boldsymbol{\mathscr{C}}$ \cite{ben}] Let $\boldsymbol{\mathscr{A}}_1,\cdots, \boldsymbol{\mathscr{A}}_m$ be independent families of random vectors, with smallest constants $\mu_1,\cdots,\mu_k$ such that
\[
\|a_k\|_{\infty}^{2}\leq \mu_k
\]
holds almost surely for $a_k\sim\boldsymbol{\mathscr{A}}_k$. The coherence of an unsaturated collection $\boldsymbol{\mathscr{C}}=\{\boldsymbol{\mathscr{A}}_k\}_{k=1}^{m}$ is 
\[
\mu(\boldsymbol{\mathscr{C}}) = \max_{k\in [m]}\mu_k.
\]
\end{definition}
In the above definition, a family $\boldsymbol{\mathscr{A}}_k$ is saturated is it consists of a single vector and a collection is unsaturated if no family in the collection is saturated. In our context, it is easy to see that the condition $\theta < 1$ avoids saturation and the definition above applies. The coherence of our collection of families will translate to the DFT-incoherence parameter defined in Section \ref{SignalModel}.

With these definitions in mind, we now state a simplified version of Theorem 13.12 in \cite{ben} that will show the $\mathcal{G}$-RIP for our ensemble.
\begin{theorem}
\label{GRIPC}
Let $0<\delta$,$\epsilon<1$,$n\geq s\geq 2$, $\boldsymbol{\mathscr{C}}=\{\boldsymbol{\mathscr{A}}_k\}_{k=1}^{m}$ be a nondegenerate collection generating the rows of $\mathcal{A}$. Suppose that
\begin{equation}
\label{Gsampling}
m \geq \frac{\tilde{c}_1\|\mathcal{G}_{\boldsymbol{\mathscr{C}}}^{-1}\|^2\mu(\boldsymbol{\mathscr{C}})s}{\delta^2}\left(\log\left(2\left(\|\mathcal{G}_{\boldsymbol{\mathscr{C}}}^{-1}\|^2\mu(\boldsymbol{\mathscr{C}})s + 1\right)\right)\log^2(s)\log(n) + \log\left(\epsilon^{-1}\right)\right),
\end{equation}
where $\tilde{c}_1$ is an absolute constant. Then with probability at least $1-\epsilon$, the matrix $\mathcal{A}$ has the $\mathcal{G}$-RIP of order $s$ with constant $\delta_{s,\mathcal{G}} \leq \delta$.
\end{theorem}

In conclusion, to obtain Theorem \ref{CS} we will first show that $\mathcal{A}$ is generated by a nondegenerate collection with unique positive-definite square root $\mathcal{G}$. Establishing this will provide a upper bounds for $\|\mathcal{G}^{-1}\|$, $\|\mathcal{G}\|$, and $\mu(\boldsymbol{\mathscr{C}})$. At this point, Theorem \ref{GRIPC} will provide  $\mathcal{A}$ with the $\mathcal{G}$-RIP and subsequently Theorem \ref{GRIP} can be applied to obtain the error bounds.

To establish that the collection $\boldsymbol{\mathscr{C}}=\{\boldsymbol{\mathscr{A}}_k\}_{k=1}^{m}$ above is nondegenerate, it suffices to show that
\begin{equation}
\label{Ebounds}
\frac{1}{m}\mathbb{E}\|\mathcal{A}w\|_2^2 \leq \beta^2(1+\theta)\|w\|_2^2 \ \ \ \ \mbox{and} \ \ \ \ \frac{1}{m}\mathbb{E}\|\mathcal{A}w\|_2^2 \geq \alpha^2(1-\theta)\|w\|_2^2
\end{equation}
for all $w\in\mathbb{C}^{n}$. This will show that $\frac{1}{m}\mathbb{E}\mathcal{A}^*\mathcal{A}$ is positive-definite if the deviation model satisfies $\theta < 1$. Further, let $\mathcal{G}$ be the unique positive-definite square root of $\frac{1}{m}\mathbb{E}\mathcal{A}^*\mathcal{A}$, then (\ref{Ebounds}) will also show that
\begin{equation}
\label{Gbounds}
\|\mathcal{G}\| \leq \beta\sqrt{1+\theta} \ \ \mbox{and} \ \ \|\mathcal{G}^{-1}\| \leq \frac{1}{\alpha\sqrt{1-\theta}}.
\end{equation}

To this end, let $w\in\mathbb{C}^n$ and normalize $\tilde{\mathcal{N}} := \sqrt{N}\mathcal{N}$ so that for $k\in [m], \ell\in [N]$
\[
\tilde{\mathcal{N}}_{k\ell} := \textbf{e}(-\tilde{t}_k(\ell-\tilde{N}-1)).
\]
Throughout, let $\tilde{\Delta}\in\mathbb{R}$ be an independent copy of the entries of $\Delta\in\mathbb{R}^m$. Then with $v := \mathcal{F}^*\Psi w$,
\begin{align*}
& \frac{1}{m}\mathbb{E}\|\mathcal{A}w\|_2^2 = \frac{1}{m}\mathbb{E}\|\tilde{\mathcal{N}}\mathcal{F}^*\Psi w\|_2^2 := \frac{1}{m}\mathbb{E}\|\tilde{\mathcal{N}}v\|_2^2\\
&= \mathbb{E}\frac{1}{m}\sum_{k=1}^{m}\lvert\langle\tilde{\mathcal{N}}_{k*},v\rangle\rvert^2
= \mathbb{E}\frac{1}{m}\sum_{k=1}^{m}\Biggl\lvert\sum_{\ell=1}^{N}\textbf{e}(\tilde{t}_k(\ell-\tilde{n}-1))v_{\ell}\Biggr\rvert^2\\
&= \mathbb{E}\frac{1}{m}\sum_{k=1}^{m}\left(\sum_{\ell=1}^{N}\sum_{\tilde{\ell}=1}^{N}\textbf{e}(\tilde{t}_k(\ell-\tilde{\ell}))v_{\ell}\bar{v}_{\tilde{\ell}}\right)\\
&= \sum_{\ell=1}^{N}\sum_{\tilde{\ell}=1}^{N}v_{\ell}\bar{v}_{\tilde{\ell}}\left(\mathbb{E}\frac{1}{m}\sum_{k=1}^{m}\textbf{e}(\tilde{t}_k(\ell-\tilde{\ell}))\right)\\
&= \sum_{\ell=1}^{N}\lvert v_{\ell}\rvert^2 + \sum_{j= 1}^{\lfloor (N-1)/m \rfloor}\sum_{\ell-\tilde{\ell} = jm}v_{\ell}\bar{v}_{\tilde{\ell}}\mathbb{E}\textbf{e}(jm(\tilde{\Delta}-1/2)) \\ 
&+ \sum_{j= 1}^{\lfloor (N-1)/m \rfloor}\sum_{\ell-\tilde{\ell} = -jm}v_{\ell}\bar{v}_{\tilde{\ell}}\mathbb{E}\textbf{e}(-jm(\tilde{\Delta}-1/2)).
\end{align*}
The last equality can be obtained as follows,
\begin{align*}
&\mathbb{E}\frac{1}{m}\sum_{k=1}^{m}\textbf{e}(\tilde{t}_k(\ell-\tilde{\ell})) = \mathbb{E}\frac{1}{m}\sum_{k=1}^{m}\textbf{e}\left(\left(\frac{k-1}{m}-\frac{1}{2} + \Delta_k\right)(\ell-\tilde{\ell})\right) \\
&= \frac{1}{m}\sum_{k=1}^{m}\textbf{e}\left(\left(\frac{k-1}{m}-\frac{1}{2}\right)(\ell-\tilde{\ell})\right)\mathbb{E}\textbf{e}(\Delta_k(\ell-\tilde{\ell}))\\ 
&= \frac{1}{m}\sum_{k=1}^{m}\textbf{e}\left(\left(\frac{k-1}{m}-\frac{1}{2}\right)(\ell-\tilde{\ell})\right)\mathbb{E}\textbf{e}(\tilde{\Delta}(\ell-\tilde{\ell}))\\
&= \mathbb{E}\textbf{e}\left((\tilde{\Delta}-1/2)(\ell-\tilde{\ell})\right)\sum_{k=1}^{m}\frac{1}{m}\textbf{e}\left(\frac{k-1}{m}(\ell-\tilde{\ell})\right) \\ 
&=
\left\{
	\begin{array}{ll}
		1  & \mbox{if } \ell = \tilde{\ell} \\
		\mathbb{E}\textbf{e}\left(jm(\tilde{\Delta}-1/2)\right)  & \mbox{if } \ell - \tilde{\ell} = jm, j\in \mathbb{Z}/\{0\} \\
		0 & \mbox{otherwise.} 
	\end{array}
\right.
\end{align*}
The third equality uses the fact that $\mathbb{E}\textbf{e}(\Delta_k(\ell-\tilde{\ell})) = \mathbb{E}\textbf{e}(\tilde{\Delta}(\ell-\tilde{\ell}))$ for all $k\in [m]$ in order to properly factor out this constant from the sum in the fourth equality. The last equality is due to the geometric series formula. 

Returning to our original calculation, we bound the last term using our deviation model assumptions
\begin{align*}
&\Biggl\lvert\sum_{j= 1}^{\lfloor (N-1)/m \rfloor}\sum_{\ell-\tilde{\ell} = -jm}v_{\ell}\bar{v}_{\tilde{\ell}}\mathbb{E}\textbf{e}(-jm(\tilde{\Delta}-1/2))\Biggr\rvert \\
&= \Biggl\lvert\sum_{j= 1}^{\lfloor (N-1)/m \rfloor}\sum_{\ell\in Q_j}v_{\ell}\bar{v}_{\ell+jm}\mathbb{E}\textbf{e}(-jm(\tilde{\Delta}-1/2))\Biggr\rvert \\
&\leq \sum_{j= 1}^{\lfloor (N-1)/m \rfloor}\sum_{\ell\in Q_j}\lvert v_{\ell}\rvert \lvert v_{\ell+jm}\rvert \lvert\mathbb{E}\textbf{e}(-jm(\tilde{\Delta}-1/2))\rvert \\
&\leq \frac{\theta m}{2N}\sum_{j= 1}^{\lfloor (N-1)/m \rfloor}\sum_{\ell\in Q_j}\lvert v_{\ell}\rvert \lvert v_{\ell + jm}\rvert \leq \frac{\theta m}{2N}\sum_{j= 1}^{\lfloor (N-1)/m \rfloor}\lvert\langle v,v\rangle\rvert \\ 
&= \frac{\theta m \|v\|_2^2}{2N}\sum_{j= 1}^{\lfloor (N-1)/m \rfloor}1 = \frac{\theta m \|v\|_2^2}{2N}\left\lfloor{ \frac{N-1}{m}} \right\rfloor \leq \frac{\theta \|v\|_2^2}{2} .
\end{align*}
$Q_j\subset [N]$ is the index set of allowed $\ell$ indices according to $j$, i.e., that satisfy $\ell \in [N]$ and $\ell+jm \in [N]$. The second inequality holds by our deviation model assumption (\ref{thetadef}).

The remaining sum (with $\ell-\tilde{\ell} = jm$) can be bounded similarly. Combine these inequalities with the singular values of $\Psi$ to obtain
\[
\frac{1}{m}\mathbb{E}\|\mathcal{A}w\|_2^2 \leq \|v\|_2^2 + \frac{2\theta\|v\|_2^2}{2} := \|\Psi w\|_2^2\left(1 + \theta\right) \leq \beta^2\|w\|_2^2\left(1 + \theta\right),
\]
and
\[
\frac{1}{m}\mathbb{E}\|\mathcal{A}w\|_2^2 \geq \alpha^2\|w\|_2^2\left(1 - \theta\right).
\]
We will apply this inequality and similar orthogonality properties in what follows (e.g., in Section \ref{proofOS}), and ask the reader to keep this in mind. 

To upper bound the coherence of the collection $\boldsymbol{\mathscr{C}}$, let $\tilde{\mathcal{N}}_{k*}\mathcal{F}^*\Psi\sim \boldsymbol{\mathscr{A}}_k$ as above. Then
\begin{align*}
&\big\|\tilde{\mathcal{N}}_{k*}\mathcal{F}^*\Psi\big\|_{\infty} = \max_{\ell\in [n]}\big\lvert\langle\tilde{\mathcal{N}}_{k*},(\mathcal{F}^*\Psi)_{*\ell}\rangle \big\rvert \\ 
&\leq\max_{\ell\in[n]}\big\|\tilde{\mathcal{N}}_{k*}\big\|_{\infty}\|(\mathcal{F}^*\Psi)_{*\ell}\|_1 = \max_{\ell\in[n]}\sum_{k=1}^{N}\rvert\langle\mathcal{F}_{*k},\Psi_{*\ell}\rangle\lvert := \gamma 
\end{align*}
and therefore
\begin{equation}
\label{Ccoherence}
\mu(\boldsymbol{\mathscr{C}}) \leq \gamma^2.
\end{equation}

The proof of Theorem \ref{CS} is now an application of Theorems \ref{GRIPC} and \ref{GRIP} using the derivations above.

\begin{proof}[Proof of Theorem \ref{CS}]
We are considering the equivalent program
\[
g^{\sharp}\coloneqq\argmin_{h\in\mathbb{C}^n}\lambda\|h\|_1+\frac{1}{\sqrt{m}}\|\mathcal{A} h - \sqrt{N}b\|_2.
\]

From the arguments above, the rows of $\mathcal{A}$ are generated by a nondegenerate collection $\boldsymbol{\mathscr{C}}$ with coherence bounded as (\ref{Ccoherence}). The unique positive-definite square root of $\frac{1}{m}\mathbb{E}\mathcal{A}^*\mathcal{A}$, denoted $\mathcal{G}$, satisfies the bounds (\ref{Gbounds}).

We now apply Theorem \ref{GRIPC} with $\delta = 1/2$, $\epsilon = n^{-1}$ and order
\[
q = 2\left \lceil{4s\left(\frac{1+\delta}{1-\delta}\right)\|\mathcal{G}\|^2\|\mathcal{G}^{-1}\|^2}\right\rceil .
\]
By (\ref{Gbounds}) and (\ref{Ccoherence}), if
\begin{equation}
\label{Sproof}
m \geq \frac{\tilde{c}_1\gamma^2 q}{\delta^2\alpha^2(1-\theta)}\left(\log\left(2\left(\frac{\gamma^2 q}{\alpha^2(1-\theta)} + 1\right)\right)\log^2(q)\log(n) + \log\left(n\right)\right),
\end{equation}
then (\ref{Gsampling}) is satisfied and the conclusion of Theorem \ref{GRIPC} holds. Therefore, with probability exceeding $1-n^{-1}$, $\mathcal{A}$ has $\mathcal{G}$-RIP of order $q$ with constant $\delta_{q,\mathcal{G}}\leq \delta = 1/2$. 

To show that our sampling assumption (\ref{samplecomplexityCS}) satisfies (\ref{Sproof}), notice that by (\ref{Gbounds})
\[
 q = 2\left \lceil{12s\|\mathcal{G}\|^2\|\mathcal{G}^{-1}\|^2}\right\rceil \leq 2\left \lceil{12s\frac{\beta^2(1+\theta)}{\alpha^2(1-\theta)}}\right\rceil \leq 2\left(12\left(1+\frac{1}{24}\right)s\frac{\beta^2(1+\theta)}{\alpha^2(1-\theta)}\right) \coloneqq \tilde{q}.
\]
The last inequality holds since 
\[
\frac{12s\beta^2(1+\theta)}{\alpha^{2}(1-\theta)}\geq 12s\geq 24,
\]
and for any real number $a\geq 24$ it holds that $\lceil a\rceil \leq (1+\frac{1}{24})a$. In (\ref{Sproof}), replace $q$ with $\tilde{q}$. This provides our assumed sampling complexity, where expression (\ref{samplecomplexityCS}) simplifies by absorbing all absolute constants into $C_1$ and $C_2$. 

With parameter $\lambda$ chosen for (\ref{l1minfull}), the conditions of Theorem \ref{GRIP} hold with $\delta = 1/2$ and we obtain the error bound 
\[
\|g-g^{\sharp}\|_2 \leq \frac{8\epsilon_s(g)}{\sqrt{s}} + 8\left(\frac{1}{2\lambda\sqrt{s}}+\frac{\sqrt{2}}{\alpha\sqrt{1-\theta}}\right)\frac{\sqrt{N}}{\sqrt{m}}\|\mathcal{S}f - b\|_2.
\]
To finish, notice that
\[
\|g-g^{\sharp}\|_2 \geq \frac{1}{\beta}\|\Psi(g-g^{\sharp})\|_2 = \frac{1}{\beta}\|f-\Psi g^{\sharp}\|_2, 
\]
and
\[
\|\mathcal{S}f - b\|_2 \leq \|\mathcal{S}f-\tilde{f}\|_2 + \|d\|_2 \leq 2\sqrt{m}\sum_{\lvert \ell\rvert > \frac{N-1}{2}}\lvert c_{\ell}\rvert + \|d\|_2,
\]
where the last inequality holds by Theorem \ref{thmS}.
\end{proof}

To obtain Theorem \ref{CSsimple} from Theorem \ref{CS}, notice that in Theorem \ref{CSsimple} we have $n=N$ and $\alpha=\beta=1$. The assumption $m\leq N$ gives that
\[
N \geq \frac{\gamma^2(1+\theta)s}{(1-\theta)^2} \geq \frac{(1+\theta)s}{(1-\theta)^2},
\]
which allows further simplification by combining all the logarithmic factors into a single $\polylog(N)$ term (introducing absolute constants where necessary). We note that the condition $m\leq N$ is not needed and is only applied for ease of exposition in the introductory result.

\subsection{Proof of Theorem \ref{OSthm}}
\label{proofOS}

To establish the claim, we aim to show that
\begin{equation}
\label{conditionOS}
\inf_{v\in S^{N-1}}\|\mathcal{S} v\|_2 \geq \delta>0,
\end{equation}
holds with high probability. By optimality of $f^{\sharp}$, this will give
\begin{align*}
&\|f - f^{\sharp}\|_2 \leq \frac{1}{\delta}\|\mathcal{S}(f - f^{\sharp})\|_2 \leq \frac{1}{\delta}\|\mathcal{S}f - b\|_2 + \frac{1}{\delta}\|b - \mathcal{S}f^{\sharp}\|_2 \\
&\leq \frac{2}{\delta}\|\mathcal{S}f - b\|_2 \leq \frac{2}{\delta}\|d\|_2 + \frac{4}{\delta} \sqrt{m}\sum_{\lvert\ell\rvert>\frac{N-1}{2}}\lvert c_{\ell}\rvert,
\end{align*}
where the last inequality is due to our noise model and trigonometric interpolation error (Theorem \ref{thmS}).

To this end, we normalize by letting $\tilde{\mathcal{S}} = \frac{1}{\sqrt{m}}\tilde{\mathcal{N}}\mathcal{F}^* := \frac{\sqrt{N}}{\sqrt{m}}\mathcal{N}\mathcal{F}^*$ and note that when $m\geq N$ our sampling operator is isometric in the sense that 
\begin{equation}
\label{isotropyOS}
\mathbb{E}\tilde{\mathcal{S}}^*\tilde{\mathcal{S}} =  \mathcal{F}\left(\frac{1}{m}\mathbb{E}\tilde{\mathcal{N}}^*\tilde{\mathcal{N}}\right)\mathcal{F}^* = \mathcal{I}_{N}
\end{equation}
where $\mathcal{I}_{N}$ is the $N\times N$ identity matrix. To see this, we use our calculations from the previous section (that establish (\ref{Ebounds})) to obtain as before that for $\ell,\tilde{\ell}\in [N]$
\begin{align*}
&\mathbb{E}\left(\frac{1}{m}\tilde{\mathcal{N}}^*\tilde{\mathcal{N}}\right)_{\ell\tilde{\ell}} = \frac{1}{m}\mathbb{E}\langle \tilde{\mathcal{N}}_{*\ell},\tilde{\mathcal{N}}_{*\tilde{\ell}}\rangle = \frac{1}{m}\mathbb{E}\sum_{k=1}^{m}\textbf{e}(\tilde{t}_k(\ell-\tilde{\ell}))\\
&=
\left\{
	\begin{array}{ll}
		1  & \mbox{if } \ell = \tilde{\ell} \\
		\mathbb{E}\textbf{e}\left(jm(\tilde{\Delta}-1/2)\right)  & \mbox{if } \ell - \tilde{\ell} = jm, j\in \mathbb{Z}/\{0\} \\
		0 & \mbox{otherwise.} 
	\end{array}
\right.
\end{align*}
However, if $m\geq N$, notice that the middle case never occurs since $\lvert\ell-\tilde{\ell}\rvert\leq N-1 < m$ for all $\ell,\tilde{\ell}\in [N]$. Therefore, (\ref{isotropyOS}) holds.

With the isometry established, we may now proceed to the main component of the proof of Theorem \ref{OSthm}.

\begin{theorem}
\label{RIPOS}
Let $m\geq \kappa N$ with $\kappa \geq \frac{2\log(N)}{\log(\sqrt{e}/\sqrt{2})}$ and the entries of $\Delta$ be i.i.d. with any distribution. Then
\[
\inf_{v\in S^{N-1}}\|\mathcal{S} v\|_2 \geq \frac{\sqrt{m}}{\sqrt{2N}},
\]
with probability exceeding $1-\frac{1}{N}$.

\end{theorem}

\begin{proof}[Proof:]
We will apply a matrix Chernoff inequality to lower bound the smallest eigenvalue of $\tilde{\mathcal{S}}^*\tilde{\mathcal{S}}$. To apply Theorem 1.1 in \cite{user}, notice that we can expand
\[
\tilde{\mathcal{S}}^*\tilde{\mathcal{S}} = \sum_{k=1}^{m}\tilde{\mathcal{S}}_{k*}^*\tilde{\mathcal{S}}_{k*},
\]
which is a sum of independent, random, self-adjoint, and positive-definite matrices. Our isometry condition (\ref{isotropyOS}) gives that $\mathbb{E}\tilde{\mathcal{S}}^*\tilde{\mathcal{S}} =\mathcal{I}_N$ has extreme eigenvalues equal to 1, we stress that this holds because we assume $m\geq N$ as shown above. Further,
\[
\|\tilde{\mathcal{S}}_{k*}^*\tilde{\mathcal{S}}_{k*}\| = \frac{1}{m} \|\mathcal{F}\tilde{\mathcal{N}}_{k*}^*\tilde{\mathcal{N}}_{k*}\mathcal{F}^*\| = \frac{1}{m} \|\tilde{\mathcal{N}}_{k*}^*\tilde{\mathcal{N}}_{k*}\| = \frac{N}{m}.
\]
Therefore, by Theorem 1.1 in \cite{user} with $R=\frac{N}{m}$ and $\delta = \frac{1}{2}$, we obtain
\[
\mathbb{P}\left(\lambda_{\mbox{min}}\left(\tilde{\mathcal{S}}^*\tilde{\mathcal{S}}\right)\leq \frac{1}{2} \right) \leq N\left(\frac{\sqrt{2}}{\sqrt{e}}\right)^{m/N}.
\]
With $m\geq \kappa N$ and $\kappa \geq \frac{2\log(N)}{\log(\sqrt{e}/\sqrt{2})}$, the left hand side is upper bounded by $N^{-1}$. Since the singular values of $\tilde{\mathcal{S}}$ are the squareroot of the eigenvalues of $\tilde{\mathcal{S}}^*\tilde{\mathcal{S}}$, this establishes the result.
\end{proof}

With our remarks in the beginning of the section, we can now easily establish the proof of Theorem \ref{OSthm}.

\begin{proof}[Proof of Theorem \ref{OSthm}:]
Under our assumptions, apply Theorem \ref{RIPOS} to obtain that for all $v\in S^{N-1}$
\[
\inf_{v\in S^{N-1}}\|\mathcal{S}v\|_2 \geq \frac{\sqrt{m}}{\sqrt{2N}}
\]
holds with the prescribed probability. This establishes (\ref{conditionOS}) with $\delta = \frac{\sqrt{m}}{\sqrt{2N}}$. The remainder of the proof follows from our outline in the beginning of the section.
\end{proof}

\subsection{Interpolation Error of Dirichlet Kernel: Proof}
\label{errorproof}

In this section we provide the error term of our interpolation operator when applied to our signal model (Theorem \ref{thmS}) and also the error bound given in Corollary \ref{fullsignalerror}.

\begin{proof}[Proof of Theorem \ref{thmS}:]
We begin by showing (\ref{perfectinterpolation}), i.e., if $\tilde{t}_k = t_{\tilde{p}}$ for some $\tilde{p}\in [N]$ (our ``nonuniform'' sample lies on the equispaced interpolation grid) then the error is zero. This is easy to see by orthogonality of the complex exponentials, combining (\ref{fft}), (\ref{nfft}) (recall that $\tilde{N} = \frac{N-1}{2}$) we have
\begin{align*}
&(\mathcal{S}f)_{k} = \langle f, \mathcal{S}_{k*}\rangle = \sum_{p=1}^{N} f_p\mathcal{S}_{kp} = \frac{1}{N}\sum_{p=1}^{N} f_p\left(\sum_{u=-\tilde{N}}^{\tilde{N}}\textbf{e}(ut_p)\textbf{e}(-u\tilde{t}_k)\right) \\ 
&= \frac{1}{N}\sum_{p=1}^{N} f_p\left(\sum_{u=-\tilde{N}}^{\tilde{N}}\textbf{e}(ut_p)\textbf{e}(-ut_{\tilde{p}})\right) = f_{\tilde{p}} = \mathbf{f}(t_{\tilde{p}}) = \mathbf{f}(\tilde{t}_k) = \tilde{f}_k.
\end{align*}
The fourth equality holds since we are assuming $\tilde{t}_k = t_{\tilde{p}}$ for some $\tilde{p}\in [N]$.

We now deal with the general case (\ref{interpolationerror}). Recall the Fourier expansion of our underlying function
\[
\mathbf{f}(x) = \sum_{\ell=-\infty}^{\infty}c_{\ell}\mathbf{e}(\ell x).
\]
Again, using (\ref{fft}), (\ref{nfft}) and the Fourier expansion at $\textbf{f}(t_p) = f_p$ we obtain
\begin{align*}
&(\mathcal{S}f)_{k} = \langle f, \mathcal{S}_{k*}\rangle = \sum_{p=1}^{N} f_p\mathcal{S}_{kp}\\
&:= \frac{1}{N}\sum_{p=1}^{N} \left(\sum_{\ell=-\infty}^{\infty}c_{\ell}\mathbf{e}(\ell t_p)\right)\left(\sum_{u=-\tilde{N}}^{\tilde{N}}\textbf{e}(ut_p)\textbf{e}(-u\tilde{t}_k)\right). \\
\end{align*}
At this point, we wish to switch the order of summation and sum over all $p\in[N]$. We must assume the corresponding summands are non-zero. To this end, we continue assuming $f_p, \mathcal{S}_{kp} \neq 0$ for all $p\in [N]$. We will deal with these cases separately afterward. In particular we will remove this assumption for the $f_p$'s and show that $\mathcal{S}_{kp} \neq 0$ under our assumption $\tilde{\tau}\subset\Omega$. 

Proceeding, we may now sum over all $p\in [N]$ to obtain
\begin{align*}
(\mathcal{S}f)_{k} &= \frac{1}{N}\sum_{u=-\tilde{N}}^{\tilde{N}}\sum_{\ell=-\infty}^{\infty}c_{\ell}\textbf{e}(-u\tilde{t}_k)\sum_{p=1}^{N}\textbf{e}((u + \ell)t_p)\\
&= \sum_{u=-\tilde{N}}^{\tilde{N}}\sum_{j=-\infty}^{\infty}(-1)^{jN}c_{jN+u}\textbf{e}(u\tilde{t}_k)
= \sum_{j=-\infty}^{\infty}(-1)^{\lfloor\frac{j+\tilde{N}}{N}\rfloor}c_j\textbf{e}(r(j)\tilde{t}_k).
\end{align*}
The second equality is obtained by orthogonality of the exponential basis functions, $\sum_{p=1}^{N}\textbf{e}((u + \ell)t_p) = 0$ when $\ell+u\notin N\mathbb{Z}$ and otherwise equal to $N(-1)^{jN}$ for some $j\in\mathbb{Z}$ where $u+\ell = jN$. The last equality results from a reordering of the absolutely convergent series where the mapping $r$ is defined as in the statement of Theorem \ref{thmS}. 

To illustrate the reordering, we consider $j\geq 0$ (for simplicity) and first notice that $(-1)^{jN} = (-1)^j$ since $N$ is assumed to be odd in Section \ref{SignalModel}. Aesthetically expanding the previous sum gives
\begin{align*}
\sum_{u=-\tilde{N}}^{\tilde{N}}\sum_{j=0}^{\infty} (-1)^j c_{jN+u}\textbf{e}\left(u\tilde{t}_k\right)& = \\
\hspace{15pt} \textbf{e}(-\tilde{N}\tilde{t}_k)&\left(c_{-\tilde{N}} \hspace{10pt} -c_{N-\tilde{N}}\hspace{10pt} + c_{2N-\tilde{N}}\hspace{10pt} - \ldots\right)\\
+ \hspace{5pt} \textbf{e}((-\tilde{N}+1)\tilde{t}_k)&\left(c_{-\tilde{N}+1} -c_{N-\tilde{N}+1} + c_{2N-\tilde{N}+1} - \ldots \right)\\
&\vdots\\
+ \hspace{33pt} \textbf{e}(0\cdot\tilde{t}_k)&\left(\hspace{1pt}c_{0}\hspace{20pt} -c_{N}\hspace{24pt} + c_{2N} \hspace{23pt} - \ldots\right)\\
&\vdots\\
+ \hspace{36pt}\textbf{e}(\tilde{N}\tilde{t}_k)&\left(c_{\tilde{N}}\hspace{16pt} -c_{N+\tilde{N}}\hspace{11pt} + c_{2N+\tilde{N}} \hspace{10pt} - \ldots\right).
\end{align*}
Notice that in the first row starting at the second coefficient we have indices $N-\tilde{N} = \tilde{N}+1$ followed by $2N-\tilde{N} = N+\tilde{N}+1$ and so on, which are subsequent to the indices of the coefficients in the last row (one column prior). Therefore, if start at the top left coefficient $c_{-\tilde{N}}$ and ``column-wise'' traverse this infinite array of Fourier coefficients we will obtain the ordered sequence $\{(-1)^{\lfloor\frac{j+\tilde{N}}{N}\rfloor}c_j\}_{j=-\tilde{N}}^{\infty}$ (with no repetitions). 

The coefficients in row $q\in[N]$ correspond to frequency value $-\tilde{N}+q-1$ and have indices of the form $pN-\tilde{N}+q-1$ for some $p\in\mathbb{N}$. To establish that the reordered series is equivalent, we finish by checking that for a given index the mapping $r$ gives the correct frequency value, i.e., $r(pN-\tilde{N}+q-1) = -\tilde{N}+q-1$ for all $q\in[N]$:
\begin{align*}
&r(pN-\tilde{N}+q-1) \coloneqq \mbox{rem}(pN-\tilde{N}+q-1 + \tilde{N}, N) - \tilde{N}\\
&= \mbox{rem}(pN+q-1, N) - \tilde{N} = q-1 - \tilde{N}.
\end{align*}
We can therefore reorder the series as desired and incorporate the sum over $j<0$ via the same logic to establish the equality. 

Since for $\ell\in\{-\tilde{N}, -\tilde{N} + 1, \cdots, \tilde{N}\}$ we have $r(\ell) = \ell$ and $(-1)^{\lfloor\frac{\ell+\tilde{N}}{N}\rfloor} = 1$, we finally obtain
\begin{align*}
&\textbf{f}(\tilde{t}_k)-(\mathcal{S}f)_{k}  = \sum_{\lvert\ell\rvert>\tilde{N}}c_{\ell}\left(\textbf{e}(\ell\tilde{t}_k) - (-1)^{\lfloor\frac{\ell+\tilde{N}}{N}\rfloor}\textbf{e}(r(\ell)\tilde{t}_k)\right). 
\end{align*}
The definition of the $p$-norms along with the triangle inequality give the remaining claim. In particular,
\begin{align*}
&\|\tilde{f}-\mathcal{S}f\|_p = \left(\sum_{k=1}^{m} \lvert\textbf{f}(\tilde{t}_k)-(\mathcal{S}f)_{k}\rvert^p\right)^{1/p}\\
& = \left(\sum_{k=1}^{m} \Biggl\lvert\sum_{\lvert\ell\rvert>\tilde{N}}c_{\ell}\left(\textbf{e}(\ell\tilde{t}_k) - (-1)^{\lfloor\frac{\ell+\tilde{N}}{N}\rfloor}\textbf{e}(r(\ell)\tilde{t}_k)\right)\Biggr\rvert^p\right)^{1/p}\\
&\leq \left(\sum_{k=1}^{m} \left(\sum_{\lvert\ell\rvert>\tilde{N}}2\lvert c_{\ell}\rvert\right)^p\right)^{1/p} 
 = \left(m \left(\sum_{\lvert\ell\rvert>\tilde{N}}2\lvert c_{\ell}\rvert\right)^p\right)^{1/p} = 2m^{1/p} \sum_{\lvert\ell\rvert>\tilde{N}}\lvert c_{\ell}\rvert.
\end{align*}

This finishes the proof in the case $f_p, \mathcal{S}_{kp} \neq 0$ for all $p\in [N]$. To remove this condition for the $f_p$'s, we may find a real number $\mu$ such that the function
\[
\mathbf{g}(x) \coloneqq \mathbf{f}(x) + \mu = \sum_{\ell\in (-\infty,\infty)\cap\mathbb{Z}/\{0\}}c_{\ell}\mathbf{e}(\ell x) + c_0 + \mu
\]
is non-zero when $x \in \{t_p\}_{p=1}^{N}$. In particular notice that if we define $h = f + \mu \in \mathbb{C}^N$, then $h_p \neq 0$ for all $p\in [N]$. Therefore, only assuming now that $\mathcal{S}_{kp} \neq 0$ for $p\in [N]$, the previous argument can be applied to conclude
\begin{align*}
&\textbf{g}(\tilde{t}_k)-(\mathcal{S}h)_{k}  = \sum_{\lvert\ell\rvert>\tilde{N}}c_{\ell}\left(\textbf{e}(\ell\tilde{t}_k) - (-1)^{\lfloor\frac{\ell+\tilde{N}}{N}\rfloor}\textbf{e}(r(\ell)\tilde{t}_k)\right).
\end{align*}
However, if $1_N\in\mathbb{C}^{N}$ denotes the all ones vector and $e_{\tilde{N}+1}\in\mathbb{C}^{N}$ is the $\tilde{N}+1$-th standard basis vector, notice that
\begin{align*}
&(\mathcal{S}h)_{k} = \langle \mathcal{S}_{k*},h\rangle = \langle \mathcal{S}_{k*},f\rangle + \mu\langle \mathcal{S}_{k*},1_N\rangle = \langle \mathcal{S}_{k*},f\rangle + \mu\langle \mathcal{N}_{k*},\mathcal{F}^*1_N\rangle \\ 
&= \langle \mathcal{S}_{k*},f\rangle + \mu\sqrt{N}\langle \mathcal{N}_{k*}, e_{\tilde{N}+1}\rangle = \langle \mathcal{S}_{k*},f\rangle + \mu = (\mathcal{S}f)_{k} + \mu.
\end{align*}
The fourth equality holds by orthogonality of $\mathcal{F}^*$ and since $\mathcal{F}_{(\tilde{N}+1)*}^* = \frac{1}{\sqrt{N}}1_N$. The fifth inequality holds since $\mathcal{N}_{k(\tilde{N}+1)}=\frac{1}{\sqrt{N}}$. Therefore
\[
\textbf{g}(\tilde{t}_k)-(\mathcal{S}h)_{k} = \textbf{f}(\tilde{t}_k)+\mu-\left((\mathcal{S}f)_{k} + \mu\right) = \textbf{f}(\tilde{t}_k)-(\mathcal{S}f)_{k},
\]
and the claim holds in this case as well.

The assumption $\mathcal{S}_{kp} \neq 0$ will always hold if $\tilde{\tau}\subset\Omega$, i.e., $\tilde{t}_k\in[-\frac{1}{2},\frac{1}{2})$ for all $k\in [m]$. We show this case by deriving conditions under which this occurs. As noted before, we have
\begin{align*}
&\mathcal{S}_{kp} := \sum_{u=-\tilde{N}}^{\tilde{N}}\textbf{e}(u(t_p-\tilde{t}_k)) = \sum_{u=0}^{N-1}\textbf{e}(u(t_p-\tilde{t}_k))\textbf{e}(-\tilde{N}(t_p-\tilde{t}_k)) \\ 
&= \textbf{e}(-\tilde{N}(t_p-\tilde{t}_k))\frac{1-\textbf{e}(N(t_p-\tilde{t}_k))}{1-\textbf{e}(t_p-\tilde{t}_k)}
\end{align*}
and we see that $\mathcal{S}_{kp} = 0$ iff $N(t_p-\tilde{t}_k)\in\mathbb{Z}/\{0\}$ and $t_p-\tilde{t}_k\notin\mathbb{Z}$. However, notice that
\[
N(t_p-\tilde{t}_k) = N\left(\frac{p-1}{N}-\frac{k-1}{m}-\Delta_k\right) = p-1-\frac{N(k-1)}{m}-N\Delta_k,
\]
so that $N(t_p-\tilde{t}_k)\in\mathbb{Z}/\{0\}$ iff $\frac{N(k-1)}{m}+N\Delta_k = N\tilde{t}_k + \frac{N}{2} \in\mathbb{Z}/\{p-1\}$. This condition equivalently requires $\tilde{t}_k = \frac{j}{N} - \frac{1}{2}$ for some $j\in\mathbb{Z}/\{p-1\}$. Since this must hold for all $p\in [N]$, we have finally have that 
\[
N(t_p-\tilde{t}_k)\in\mathbb{Z}/\{0\} \ \ \ \mbox{iff} \ \ \ \tilde{t}_k = \frac{j}{N} - \frac{1}{2} \ \ \mbox{for some} \ \ j\in\mathbb{Z}/\{0,1, \cdots, N-1\}.
\]

We see that such a condition would imply that $\tilde{t}_k \notin \Omega := [-\frac{1}{2},\frac{1}{2})$, which violates our assumption $\tilde{\tau}\subset\Omega$. This finishes the proof.

\end{proof}

We end this section with the proof of Corollary \ref{fullsignalerror}.

\begin{proof}[Proof of Corollary \ref{fullsignalerror}:]
The proof will consist of applying Theorem \ref{CSsimple} (under identical assumptions) and Theorem \ref{thmS}.

By Theorem \ref{CSsimple}, we have that
\[
\|f-\Psi g^{\sharp}\|_{2} \leq \frac{8\epsilon_s(g)}{\sqrt{s}} + \left(\frac{4}{\lambda\sqrt{s}}+\frac{8\sqrt{2}}{\sqrt{1-\theta}}\right)\left(\frac{\sqrt{N}}{\sqrt{m}}\|d\|_2 + 2\sqrt{N}\sum_{\lvert\ell\rvert>\frac{N-1}{2}}\lvert c_{\ell}\rvert\right)
\]
with probability exceeding $1-\frac{1}{N}$. As in the proof of Theorem \ref{thmS}, we can show that for $x\in\Omega$
\[
\textbf{f}(x) - \langle \textbf{h}(x),\mathcal{F}^*f\rangle = \sum_{\lvert\ell\rvert>\tilde{N}}c_{\ell}\left(\textbf{e}(\ell x) - (-1)^{\lfloor\frac{\ell+\tilde{N}}{N}\rfloor}\textbf{e}(r(\ell)x)\right).
\]
Therefore
\begin{align*}
&\lvert\textbf{f}(x)-\textbf{f}^{\sharp}(x)\rvert := \lvert\textbf{f}(x)-\langle \textbf{h}(x),\mathcal{F}^*\Psi g^{\sharp}\rangle\rvert \\ 
&\leq \lvert\textbf{f}(x)-\langle \textbf{h}(x),\mathcal{F}^*f\rangle\rvert + \lvert\langle \textbf{h}(x),\mathcal{F}^*f\rangle-\langle \textbf{h}(x),\mathcal{F}^*\Psi g^{\sharp}\rangle\rvert \\
&\leq \Biggl\lvert\sum_{\lvert\ell\rvert>\tilde{N}}c_{\ell}\left(\textbf{e}(\ell x) - (-1)^{\lfloor\frac{\ell+\tilde{N}}{N}\rfloor}\textbf{e}(r(\ell)x)\right)\Biggr\rvert + \|\textbf{h}(x)\|_2\|\mathcal{F}^*(f-\Psi g^{\sharp})\|_2\\
&\leq 2\sum_{\lvert \ell\rvert>\tilde{N}}\lvert c_{\ell}\rvert + \frac{8\epsilon_s(x)}{\sqrt{s}} + \left(\frac{4}{\lambda\sqrt{s}}+\frac{8\sqrt{2}}{\sqrt{1-\theta}}\right)\left(\frac{\sqrt{N}}{\sqrt{m}}\|d\|_2 + 2\sqrt{N}\sum_{\lvert\ell\rvert>\frac{N-1}{2}}\lvert c_{\ell}\rvert\right)
\end{align*}
The last inequality holds since $\|\textbf{h}(x)\|_2 = 1$ (here $x$ is considered fixed and $\textbf{h}(x)\in\mathbb{C}^{N}$). This finishes the proof.

\end{proof}

\end{document}